\documentclass{article}
\usepackage[utf8]{inputenc}
\usepackage{amsthm}
\usepackage{soul,color}
\usepackage{authblk}
\usepackage{multirow}
\usepackage{mathrsfs}
\usepackage{xspace}
\usepackage{amsfonts,amssymb}
\usepackage{pgf,tikz}
\usetikzlibrary{decorations}
\usetikzlibrary{decorations.markings}
\usepackage{tikz-network}
\usepackage{caption}
\usepackage{fullpage}
\usepackage{subcaption}
\usetikzlibrary{arrows,automata,positioning,fit,calc,shapes}
\usepackage[ruled,linesnumbered]{algorithm2e}

\definecolor{dgreen}{rgb}{0.0, 0.5, 0.0}
\definecolor{amber}{rgb}{1.0, 0.75, 0.0}
\usepackage{tikz}

\usetikzlibrary{mindmap,positioning}
\usetikzlibrary{graphs}
\usetikzlibrary[graphs]
\usetikzlibrary{patterns}
\usetikzlibrary{arrows,automata,positioning}
\usetikzlibrary{decorations.pathreplacing}
\usetikzlibrary{decorations}
\usetikzlibrary{shapes.geometric}
\usetikzlibrary{quotes}
\usetikzlibrary{shapes.callouts}
\usetikzlibrary{arrows.meta}
\usetikzlibrary{calc}

\usetikzlibrary{shadings,shapes.arrows,shadows}

\tikzset{faded/.style={gray,very thin}}
\tikzset{vertex/.style={draw,circle,minimum size=10pt,inner sep=0pt}}
\tikzset{novertex/.style={circle,minimum size=10pt,inner sep=0pt}}
\tikzset{blackvertex/.style={draw,circle,minimum size=10pt,inner sep=0pt, fill=black}}
\tikzset{redvertex/.style={draw,circle,minimum size=10pt,inner sep=0pt, fill=red}}
\tikzset{redvertexfaded/.style={draw,circle,faded,minimum size=10pt,inner sep=0pt, fill=red!50}}
\tikzset{greenvertex/.style={draw,circle,minimum size=10pt,inner sep=0pt, fill=green}}
\tikzset{greenvertexfaded/.style={draw,circle,faded,minimum size=10pt,inner sep=0pt, fill=green!50}}
\tikzset{bluevertex/.style={draw,circle,minimum size=10pt,inner sep=0pt, fill=blue}}
\tikzset{bluevertexfaded/.style={draw,circle,faded,minimum size=10pt,inner sep=0pt, fill=blue!50}}
\tikzset{yellowvertex/.style={draw,circle,minimum size=10pt,inner sep=0pt, fill=yellow}}
\tikzset{yellowvertexfaded/.style={draw,circle,faded,minimum size=10pt,inner sep=0pt, fill=yellow!50}}
\tikzset{boxvertex/.style={draw, rectangle, minimum width=1cm, minimum height=.5cm}}
\tikzset{crossedboxvertex/.style={draw, rectangle, minimum width=1cm, minimum height=.5cm,pattern=north east lines}}
\tikzset{edge/.style = {->,> = latex'}}

\protected\def\VerticalText#1{\leavevmode\bgroup\vbox\bgroup\xvvv#1\relax}

\def\xvvv{\afterassignment\xxvvv\let\tmp= }

\def\xxvvv{%
\ifx\tmp\@sptoken\egroup\ \vbox\bgroup\let\next\xvvv
\else\ifx\tmp\relax\egroup\egroup\let\next\relax
\else
\hbox to 1.1em{\hfill\tmp\hfill}
\let\next\xvvv\fi\fi
\next}
\tikzset{
    side by side/.style 2 args={
    line width=2pt,
    #1,
    postaction={
        clip,postaction={draw,#2}
        }
    }
}
\usepackage{complexity}

\newcommand{\TOBs}{{\sc tob}s\xspace}
\newcommand{\TSSs}{{\sc toss}s\xspace}
\newcommand{\TOB}{{\sc tob}\xspace}
\newcommand{\TSS}{{\sc toss}\xspace}
\newcommand{\TISS}{{\sc tiss}\xspace}
\newcommand{\TIB}{{\sc tib}\xspace}
\newcommand{\TIBs}{{\sc tib}s\xspace}
\newcommand{\MT}{{\textsc{mt}}}
\newcommand{\EA}{{\textsc{ea}}}
\newcommand{\ST}{{\textsc{st}}}
\newcommand{\FT}{{\textsc{ft}}}
\newcommand{\LD}{{\textsc{ld}}}
\newcommand{\MW}{{\textsc{mw}}}
\newcommand{\head}{{\rm h}}
\newcommand{\tail}{{\rm t}}
\newcommand{\len}{\ell}
\newcommand{\wait}{\text{wait}}
\newcommand{\weight}{\text{tt}}
\newcommand{\dur}{\text{dur}}

\newcommand{\myparagraph}[1]{\noindent\textbf{#1}}
\newcommand{\mc}{\mathcal}
\newcommand{\ba}{\Bar}
\newcommand{\cir}{\circlearrowleft}
\newcommand{\sub}[1]{_{_#1}}

\DeclareMathOperator*{\argmin}{arg\,min}

\newtheorem{theorem}{Theorem}[section]
\newtheorem{lemma}{Lemma}[section]
\newtheorem{proposition}{Proposition}[section]
\newtheorem{corollary}{Corollary}[section]
\theoremstyle{remark}
\newtheorem{remark}{Remark}[section]
\theoremstyle{definition}
\newtheorem{definition}{Definition}[section]
\newtheorem{problem}{Problem}[section]
\newtheorem{claim}{Claim}[section]

\begin{document}
\title{On Computing Optimal Temporal Branchings\\ and Spanning Subgraphs\thanks{Daniela Bubboloni is partially supported by GNSAGA of INdAM (Italy). Daniela Bubboloni, Costanza Catalano and Andrea Marino are partially supported by Italian PNRR CN4 Centro Nazionale per la Mobilità Sostenibile, NextGeneration EU - CUP,  B13C22001000001. Ana Silva is partially supported by: FUNCAP MLC-0191-00056.01.00/22 and PNE-0112-00061.01.00/16, CNPq 303803/2020-7 (Brazil).}}
%
%

\author[1]{Daniela Bubboloni}
\author[1]{Costanza Catalano}
\author[2]{Andrea Marino}
\author[3]{Ana Silva}
\small
\affil[1]{\small Department of Mathematics and Computer Sciences, 
University of Florence, Florence, Italy. 
daniela.bubboloni@unifi.it, costanza.catalano@unifi.it}
\affil[2]{Department of Statistics, Computer Sciences, Applications, University of Florence, Florence, Italy.
andrea.marino@unifi.it}
\affil[3]{Departamento de Matematica, Universidade Federal Do Ceara Fortaleza, Brazil. anasilva@mat.ufc.br}
\maketitle  
\normalsize
\begin{abstract}
\noindent 
In this work we extend the concept of out/in-branchings spanning the vertices of a digraph (also called directed spanning trees) to temporal graphs, which are digraphs where arcs are available only at prescribed times. While the literature has focused on minimum weight/earliest arrival time Temporal Out-Branchings (\TOB), we solve the problem for other optimization criteria. In particular, we define five different types of \TOBs based on the optimization of the travel duration (\FT-\TOB), of the departure time (\LD-\TOB), of the number of transfers (\MT-\TOB), of the total waiting time (\MW-\TOB), and of the travelling time (\ST-\TOB).
For $\textsc{d}\in \{\text{\LD,\MT,\ST}\}$, we provide necessary and sufficient conditions for the existence of a spanning $\textsc{d}$-\TOB; when it does not exist, we characterize the maximum vertex set that a $\textsc{d}$-\TOB can span. Moreover, we provide a log linear algorithm for computing such branchings.
For $\textsc{d}\in \{\text{\FT,\MW}\}$, we prove that deciding the existence of a spanning $\textsc{d}$-\TOB is \NP-complete; we also show that the same results hold for optimal temporal in-branchings.
Finally, we investigate the related problem of computing a spanning temporal subgraph with the minimum number of arcs and optimizing a chosen criterion $\textsc{d}$.
This problem turns out to be \NP-hard for any $\textsc{d}$.
The hardness results are quite surprising, as computing optimal paths between nodes can always be done in polynomial time.

\end{abstract}
\textbf{Keywords:} Temporal graph, temporal network, optimal branching, temporal branching, optimal temporal walk, temporal spanning subgraph.


\section{Introduction}
\label{sec:intro}
A temporal graph is a graph where arcs are active only at certain time instants, with a possible \emph{travelling time} indicating the time it takes to traverse an arc. There is not a unified terminology in the literature to call these objects, as they are also known as stream graphs \cite{latapy2018}, dynamic networks \cite{Ranshous2015}, temporal networks \cite{KEMPE2002}, and time-varying graphs \cite{Kuwata2009} to name a few. Important categories of temporal graphs are those of transport networks, where arcs are labeled by the times of bus/train/flight departures and arrivals \cite{dibbelt2018}, and communication networks as phone calls and emails networks, where each arc represents the interaction between two parties \cite{Tang2010}. Temporal graphs find application in a vast number of fields such as neural, ecological and social networks, distributed computing and epidemiology. We refer the reader to \cite{HOLME2012} for a survey on temporal graphs and their applications.
Walks in temporal graphs must respect the flow of time; for instance, in a public transports network a route can happen only at increasing time instants, since a person cannot catch a bus that already left. 
As a consequence, fundamental properties of static graphs, as the fact that concatenation of walks is a walk, do not necessarily hold in temporal graphs. 
This often makes temporal graphs much harder to handle: e.g.\ computing strongly connected components takes linear time in a static graph, but it is an \NP-complete problem in a temporal graph \cite{gabow1986}; the same happens to Eulerian walks \cite{Marino2023}.  
At the same time, it is often the case that classic theorems of graph theory may or may not hold for temporal graphs depending on how some concepts are translated into the temporal framework: this applies for example to Edmonds’ result on branchings \cite{marino2021,KEMPE2002} and Menger’s Theorem \cite{AKRIDA201946,KEMPE2002,Mertzios2019}. 

Figure \ref{fig:ex_temp_graph} shows an example of temporal graph. Informally speaking, a temporal graph is modeled as a multidigraph\footnote{A directed graph where multiple arcs having the same endpoints are allowed.} with no loops, such that each arc is labeled by a couple $(t_s,t_a)$, $t_s\leq t_a$, indicating, respectively, the starting time at which we can traverse the arc from the tail vertex and the arrival time at the head vertex. A temporal walk is a walk in the multidigraph where each arc of the walk must have an arrival time smaller than or equal to the starting time of the subsequent arc in the walk (for formal definitions see Section \ref{sec:preliminaries}).\\
\myparagraph{Shortest paths in temporal graphs.}
The notion of shortest path between two vertices $u$ and $v$ in static graphs can be 
generalized to temporal graphs in different ways, based on the chosen optimization criteria. For example, we may want a path from $u$ to $v$ that arrives the earliest possible (Earliest Arrival time, denoted by $\EA(u,v)$), that minimize the overall duration of the trip (Fastest Time, denoted by $\FT(u,v)$), that leaves the latest possible (Latest Departure time, denoted by $\LD(u,v)$), that takes the least number of arcs (Minimum Transfers, denoted by $\MT(u,v)$), that minimize the waiting time in the intermediate nodes (Minimal Waiting time, denoted by $\MW(u,v)$), or that minimize the sum of the traversing times of the arcs (Shortest Travelling time, denoted by $\ST(u,v)$)\footnote{These concepts are widely used in the literature (see \cite{Bentert2020,Brunelli2023,dibbelt2018,huang2015,wu2014,wu2016}), although they may appear with different names.}. Given $\textsc{d}\in \{\EA,\FT, \LD, \MT,\MW, \ST\}$, a path meeting the criteria $\textsc{d}$ for a pair of vertices $u$ and $v$ is said to realize $\textsc{d}(u,v)$.
For formal definitions see again Section \ref{sec:preliminaries}; Figure \ref{fig:ex_temp_graph} shows examples of such paths. Notice that the paths realizing $\textsc{d}(u,v)$ may not be unique. In fact $\MW(1,3)$ is realized both by the yellow and the red walk, while $\MT(1,3)$ is realized both by the blue and the green walk.
Each distance is computable in polynomial-time, as reported in Table~\ref{table:shortest_path}.  

\begin{table}[t]
\caption{Computational time of single source shortest paths in a temporal graph with $n$ vertices and $m$ arcs for the different criteria.} \label{table:shortest_path}
\label{tab:path}
\centering
\begin{tabular}{|c|c|c|c|c|c|}
    \hline
 \EA & \FT & \LD & \MT  & \MW & \ST\\ 
 \hline
$O(m)$  & $O(m\log n)$&    $O(m\log m)$ &  $O(m\log n)$   & $O(m\log m)$   & $O(m\log m)$
\\
\cite{huang2015,wu2014}& \cite{Brunelli2023}& \cite{Bentert2020} & \cite{Brunelli2023}& \cite{Bentert2020, Brunelli2023} & \cite{Bentert2020,wu2014,wu2016}  \\
    \hline
\end{tabular}
\end{table}

\myparagraph{Optimal temporal branchings.} In static directed graphs, spanning branchings are well-studied objects; they represent a minimal set of arcs that connect a special vertex, called the root, to any other vertex (out-branching), or any vertex to the root (in-branching). They are also called arborescences or spanning directed trees, since their underlying structure is a tree. 
Spanning branchings representing shortest distances are also well-studied. Their existence is guaranteed simply by the reachability of any vertex from/to the root and they can be computed in $O(m \log m)$ time by 
Dijkstra's algorithm \cite{Cormen2001}. Branchings are, to cite a few, important for engineering applications as they represent the cheapest or shortest way to reach all vertices \cite{GAO2019,li2004}, and in social networks in relation to information dissemination and spreading \cite{Amoruso2020,yue2020}.
We can similarly define spanning branchings in temporal graphs, here called spanning \TOBs (Temporal Out-Branchings) and \TIBs (Temporal In-Branchings), representing a minimal set of temporal arcs that temporally connect any vertex from/to the root. Equivalently, a \TOB is a temporal graph that has a branching as underlying graph and each vertex is temporally reachable from the root (see Section \ref{sec:tobTIB} for formal definitions and results).
This definition of \TOB has already appeared in the literature \cite{huang2015,KAMIYAMA2015321}.\footnote{We make notice that \cite{KAMIYAMA2015321} proposes it in a simplified context, while the conditions listed in the definition of \cite{huang2015} are not all necessary to describe the concept (see Lemma \ref{lem:tob}).}
In the context of urban mobility, suppose that a concert has just finished in a remote location $x$, and we want to guarantee that every person can go back home via public transports, while optimizing the number of bus/train rides. This problem can be solved by a spanning \TOB with root $x$. We also may ask this \TOB to arrive the earliest possible in every point of interest of the city, or the trips to last the shortest possible, 
or to optimize any of the distances that we have introduced before.
It is then natural to extend the notion of shortest distance branchings to the temporal framework. 
For each distance $\textsc{d}$, we call spanning $\textsc{d}$-\TOB a spanning \TOB such that, for every vertex $v$, the walk from the root to $v$ realizes the chosen distance $\textsc{d}(r,v)$. Figure \ref{fig:ex_Dtob} shows, for each distance $\textsc{d}$, a spanning $\textsc{d}$-\TOB with root $1$ of the temporal graph in Figure \ref{fig:ex_temp_graph}.\footnote{In general $\textsc{d}$-\TOBs are not unique. For instance, another spanning \MT-\TOB can be obtained from the one in Figure \ref{fig:MT_TOB} by adding the arc $(9,10)$ from vertex $2$ to vertex $3$ and by deleting the arc $(8,9)$ from vertex $5$ to vertex $3$.}
We define similarly spanning $\textsc{d}$-\TIBs. 
\begin{figure}[t]
\centering
\begin{tikzpicture}[scale=0.6]
\SetVertexStyle[FillColor=white]
\SetEdgeStyle[Color=black,LineWidth=1pt]
  \Vertex[x=0,y=0,label=1]{1}
  \Vertex[x=5,y=0,label=2]{2}
  \Vertex[x=5,y=5,label=3]{3}
  \Vertex[x=0,y=5,label=4]{4}
  \Vertex[x=2,y=2,label=5]{5}
  
  \Edge[Direct,color=red](1)(2)
  \Edge[Direct,label={(6,7)},color=blue,style=dashed](1)(2)
  \Edge[Direct,label={(1,2)},bend=40,color=amber](1)(4)
  \Edge[Direct,label={(6,9)}](4)(1)
  \Edge[Direct,label={(5,7)},color=dgreen](1)(5)
  \Edge[Direct,label={(2,4)},bend=20,color=amber](4)(5)
  \Edge[Direct,label={(8,9)},bend=20](5)(4)
  \Edge[Direct,label={(4,4)},bend=30,color=amber](5)(3)
  \Edge[Direct,label={(7,7)}](3)(4)
  \Edge[Direct,label={(9,10)},bend=-10,color=blue](2)(3)
  \Edge[Direct,label={(7,8)},color=red](2)(5)
  \Edge[Direct,label={(4,5)},bend=20,distance=0.8](4)(2)
  \Edge[Direct,bend=-40,color=red](5)(3)
  \Edge[Direct,label={(8,9)},bend=-40,color=dgreen, style=dashed](5)(3)
  
\end{tikzpicture}
\caption{A temporal graph with different temporal walks from vertex $1$ to vertex $3$, each one represented by a color (two-tone arcs belong to two walks). \textcolor{amber}{Yellow}: walk realizing $\EA(1,3)$ and $\MW(1,3)$. \textcolor{red}{Red}: walk realizing $\FT(1,3)$. \textcolor{blue}{Blue}: walk realizing both $\LD(1,3)$ and $\ST(1,3)$. \textcolor{dgreen}{Green}: walk realizing $\MT(1,3)$.  } 
\label{fig:ex_temp_graph}
\end{figure}
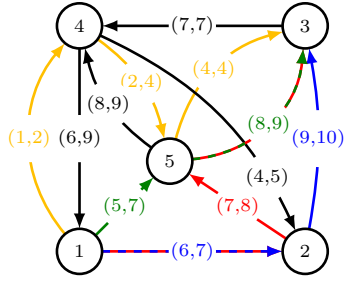

\begin{figure}[t]
\centering
\begin{subfigure}[t]{0.3\textwidth}
\begin{tikzpicture}[scale=0.4]
\SetVertexStyle[FillColor=white]
\SetEdgeStyle[Color=black,LineWidth=1pt]
  \Vertex[x=0,y=0,label=1,color=white!70!black]{1}
  \Vertex[x=5,y=0,label=2]{2}
  \Vertex[x=5,y=5,label=3]{3}
  \Vertex[x=0,y=5,label=4]{4}
  \Vertex[x=2,y=2,label=5]{5}
  \Edge[Direct,label={(1,2)},bend=40](1)(4)
  \Edge[Direct,label={(2,4)},bend=20](4)(5)
  \Edge[Direct,label={(4,4)},bend=30](5)(3)
  \Edge[Direct,label={(4,5)},bend=20, distance=0.8](4)(2)
\end{tikzpicture}
\caption{\EA-\TOB}
\label{fig:ea_tob}
    \end{subfigure}
\begin{subfigure}[t]{0.3\textwidth}
\begin{tikzpicture}[scale=0.4]
\SetVertexStyle[FillColor=white]
\SetEdgeStyle[Color=black,LineWidth=1pt]
  \Vertex[x=0,y=0,label=1,color=white!70!black]{1}
  \Vertex[x=5,y=0,label=2]{2}
  \Vertex[x=5,y=5,label=3]{3}
  \Vertex[x=0,y=5,label=4]{4}
  \Vertex[x=2,y=2,label=5]{5}
  \Edge[Direct,label={(6,7)}](1)(2)
  \Edge[Direct,label={(1,2)},bend=40](1)(4)
  \Edge[Direct,label={(7,8)},position=above right](2)(5)
  \Edge[Direct,label={(8,9)},bend=-40](5)(3)
\end{tikzpicture}
\caption{\FT-\TOB}
\end{subfigure}
   \begin{subfigure}[t]{0.3\textwidth}
\begin{tikzpicture}[scale=0.4]
\SetVertexStyle[FillColor=white]
\SetEdgeStyle[Color=black,LineWidth=1pt]
  \Vertex[x=0,y=0,label=1,color=white!70!black]{1}
  \Vertex[x=5,y=0,label=2]{2}
  \Vertex[x=5,y=5,label=3]{3}
  \Vertex[x=0,y=5,label=4]{4}
  \Vertex[x=2,y=2,label=5]{5}
  \Edge[Direct,label={(6,7)}](1)(2)
  \Edge[Direct,label={(8,9)},bend=20,position=above right](5)(4)
  \Edge[Direct,label={(9,10)},bend=-10](2)(3)
  \Edge[Direct,label={(7,8)},position=above right](2)(5)
\end{tikzpicture}
\caption{\LD-\TOB}
\end{subfigure} 
\begin{subfigure}[t]{0.3\textwidth}
\begin{tikzpicture}[scale=0.4]
\SetVertexStyle[FillColor=white]
\SetEdgeStyle[Color=black,LineWidth=1pt]
  \Vertex[x=0,y=0,label=1,color=white!70!black]{1}
  \Vertex[x=5,y=0,label=2]{2}
  \Vertex[x=5,y=5,label=3]{3}
  \Vertex[x=0,y=5,label=4]{4}
  \Vertex[x=2,y=2,label=5]{5}
  \Edge[Direct,label={(6,7)}](1)(2)
  \Edge[Direct,label={(1,2)},bend=40](1)(4)
  \Edge[Direct,label={(5,7)},position=above left](1)(5)
  \Edge[Direct,label={(8,9)},bend=-40](5)(3)
\end{tikzpicture}
\caption{\MT-\TOB}
\label{fig:MT_TOB}
\end{subfigure}
\begin{subfigure}[t]{0.3\textwidth}
\begin{tikzpicture}[scale=0.4]
\SetVertexStyle[FillColor=white]
\SetEdgeStyle[Color=black,LineWidth=1pt]
  \Vertex[x=0,y=0,label=1,color=white!70!black]{1}
  \Vertex[x=5,y=0,label=2]{2}
  \Vertex[x=5,y=5,label=3]{3}
  \Vertex[x=0,y=5,label=4]{4}
  \Vertex[x=2,y=2,label=5]{5}
  \Edge[Direct,label={(6,7)}](1)(2)
  \Edge[Direct,label={(1,2)},bend=40](1)(4)
  \Edge[Direct,label={(2,4)},bend=20](4)(5)
  \Edge[Direct,label={(4,4)},bend=30](5)(3)
\end{tikzpicture}
\caption{\MW-\TOB}
\label{fig:mw_tob}
    \end{subfigure}
\begin{subfigure}[t]{0.3\textwidth}
\begin{tikzpicture}[scale=0.4]
\SetVertexStyle[FillColor=white]
\SetEdgeStyle[Color=black,LineWidth=1pt]
  \Vertex[x=0,y=0,label=1,color=white!70!black]{1}
  \Vertex[x=5,y=0,label=2]{2}
  \Vertex[x=5,y=5,label=3]{3}
  \Vertex[x=0,y=5,label=4]{4}
  \Vertex[x=2,y=2,label=5]{5}
  \Edge[Direct,label={(6,7)}](1)(2)
  \Edge[Direct,label={(1,2)},bend=40](1)(4)
  \Edge[Direct,label={(9,10)},bend=-10](2)(3)
  \Edge[Direct,label={(7,8)},position=above right](2)(5)
\end{tikzpicture}
\caption{\ST-\TOB}
\end{subfigure}
    \caption{Example of $\textsc{d}$-\TOBs of the temporal graph in Figure (\ref{fig:ex_temp_graph}) for different distances. The grey vertex is the root of the \TOB.}
    \label{fig:ex_Dtob}
    \vspace{-0.2cm}
\end{figure}
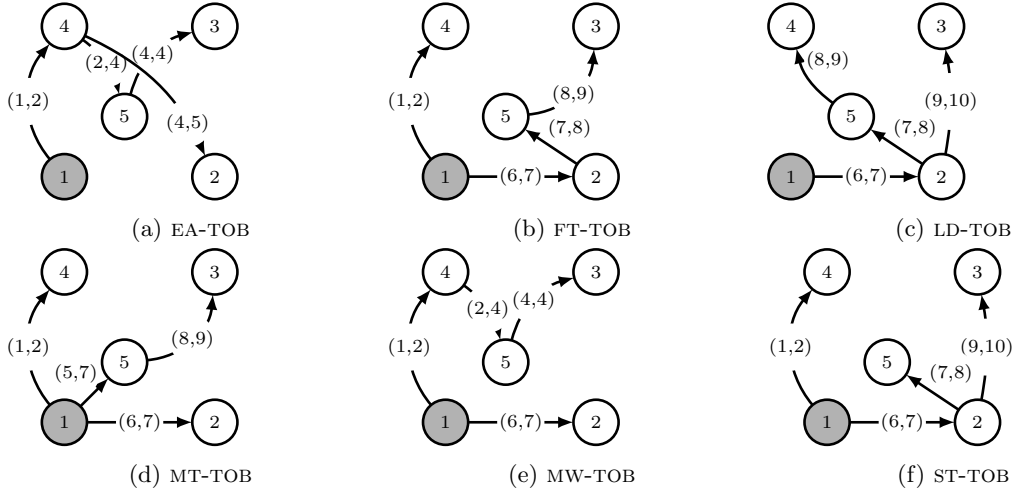

In \cite{huang2015}, the authors prove that a spanning \TOB, as well as an $\EA$-\TOB, exist 
if and only if every vertex is temporally reachable from the root. Then, they provide an algorithm to compute them in $O(m)$ time.  
As for all the other distances, the problem of computing optimal branchings is still open and seems to be a more difficult task. We start observing that for $\textsc{d}\neq \EA$, the temporal reachability from the root to any vertex is no longer sufficient for the existence of a spanning $\textsc{d}$-\TOB. That is showed in Figure \ref{fig:tob_notP} where, for each $\textsc{d}\neq \EA$, we present a temporal graph that does not admit a spanning $\textsc{d}$-\TOB even if every vertex is temporally reachable from the root. 
Indeed, in Figures (\ref{fig:no_MTLD-tob}) and (\ref{fig:no_ST-tob}) there is a unique temporal path $P$ from $r$ to $y$. Thus $P$ must be included in any spanning \TOB. Now, the particular structure of the temporal graphs under consideration implies that  $P$ is the only spanning \TOB.  
However, $P$ does not realize $\textsc{d}(r,x)$, which is realized by the temporal arc from $r $ to $x$. Therefore, $P$ is not a $\textsc{d}$-\TOB and hence no $\textsc{d}$-\TOB exists. 
In Figure (\ref{fig:no_FT-tob}), there is a unique temporal path that realizes $\textsc{d}(r,x)$, namely the one with temporal arcs $(r,v,1,1)$ and $(v,x,1,1)$. Similarly, there is a unique temporal path that realizes $\textsc{d}(r,y)$, namely the one with temporal arcs $(r,v,2,2)$ and $(v,y,2,2)$. Thus, a possible spanning $\textsc{d}$-\TOB must be equal to the graph itself, which clearly is not a branching.\footnote{
Notice that in the examples, $\tau=2$ for $\textsc{d}\in \{\FT,\LD,\MT,\MW \}$, which is the smallest value possible for which the temporal reachability from the root does not guarantee the existence of a spanning \textsc{d}-\TOB, as when $\tau=1$ the temporal graph reduces to a static graph. When $\textsc{d}=\ST$, we have that $\tau=3$: it can be proven that this is again the smallest value possible for which the temporal reachability from the root does not guarantee the existence of a spanning \ST-\TOB (Lemma \ref{claim:STtau}).} 
Notice that in all those examples, we can always find a $\textsc{d}$-\TOB on the vertex set $\{r,v,x\}$; this \TOB is highlighted by solid arcs in the figures.\footnote{In Figure (\ref{fig:no_FT-tob}) also the dashed arcs form a $\textsc{d}$-\TOB on the vertex set $\{r,v,y\}$ for $\textsc{d}\!\in\! \{\FT,\MW\}$.}
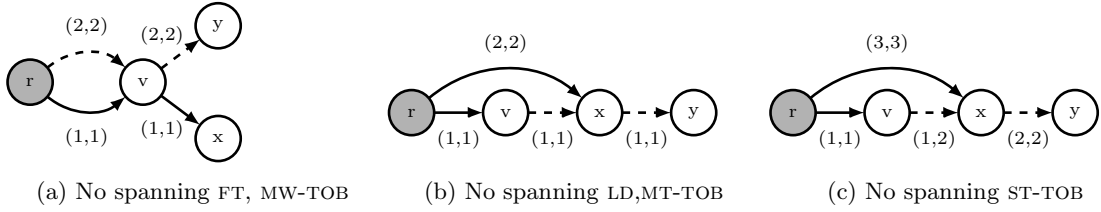
\begin{figure}[t]
\centering
\begin{subfigure}[t]{0.3\textwidth}
\begin{tikzpicture}[scale=0.5]
\SetVertexStyle[FillColor=white]
\SetEdgeStyle[Color=black,LineWidth=1pt]
  \Vertex[y=-1.5,x=0,label=r,color=white!70!black]{1}
  \Vertex[y=-1.5,x=3,label=v]{2}
  \Vertex[x=5,y=-3,label=x]{3}
  \Vertex[x=5,y=0,label=y]{4}

  \Edge[Direct,label={(1,1)},bend=-40,position=below](1)(2)
  \Edge[Direct,label={(2,2)},style=dashed,bend=40,position=above](1)(2)
  \Edge[Direct,label={(1,1)},position=below left](2)(3)
  \Edge[Direct,label={(2,2)},position=above left,style=dashed](2)(4)
\end{tikzpicture}
\caption{No spanning \FT, \MW-\TOB}
\label{fig:no_FT-tob}
\end{subfigure}
 \begin{subfigure}[t]{0.3\textwidth}
\begin{tikzpicture}[scale=0.5]
\SetVertexStyle[FillColor=white]
\SetEdgeStyle[Color=black,LineWidth=1pt]
  \Vertex[y=0,x=0,label=r,color=white!70!black]{1}
  \Vertex[y=0,x=2.5,label=v]{2}
  \Vertex[y=0,x=5,label=x]{3}
  \Vertex[y=0,x=7.5,label=y]{4}
  \Edge[Direct,label={(1,1)},position=below](1)(2)
  \Edge[Direct,label={(1,1)},style=dashed,position=below](2)(3)
  \Edge[Direct,label={(1,1)},style=dashed,position=below](3)(4)
  \Edge[Direct,label={(2,2)},bend=40,position=above](1)(3)
\end{tikzpicture}
\caption{No spanning \LD,\MT-\TOB}
\label{fig:no_MTLD-tob}
\end{subfigure}
\begin{subfigure}[t]{0.3\textwidth}
\begin{tikzpicture}[scale=0.5]
\SetVertexStyle[FillColor=white]
\SetEdgeStyle[Color=black,LineWidth=1pt]
  \Vertex[y=0,x=0,label=r,color=white!70!black]{1}
  \Vertex[y=0,x=2.5,label=v]{2}
  \Vertex[y=0,x=5,label=x]{3}
  \Vertex[y=0,x=7.5,label=y]{4}
  
  \Edge[Direct,label={(1,1)},position=below](1)(2)
  \Edge[Direct,label={(1,2)},style=dashed,position=below](2)(3)
  \Edge[Direct,label={(2,2)},style=dashed,position=below](3)(4)
  \Edge[Direct,label={(3,3)},bend=40,position=above](1)(3)
\end{tikzpicture}
\caption{No spanning \ST-\TOB}
\label{fig:no_ST-tob}
\end{subfigure}

\caption{Examples of temporal graphs that do not admit a spanning $\textsc{d}$-\TOB with root $r$. Solid arcs represent a maximum $\textsc{d}$-\TOB.} 
\label{fig:tob_notP}
\end{figure}
The following questions naturally arise:
\begin{enumerate}
    \item When does a spanning $\textsc{d}$-\TOB exist?
    \item If it does not exist, can we identify a set of vertices of maximum size that can be spanned by a $\textsc{d}$-\TOB (\emph{maximum} $\textsc{d}$-\TOB)?
    \item Can we compute a maximum $\textsc{d}$-\TOB in polynomial time? 
    \item Can we answer to all the above questions for $\textsc{d}$-\TIBs?
\end{enumerate}
We observe that having not all the vertices be spanned by a $\textsc{d}$-\TOB might be an issue in real world situations. For example, in the public transports setting, where we might want to reach anyway all the vertices of the graph (places in a city) by optimizing some distance and while still using the least amount of connections possible (buses/trams/...). We can then approach the problem from another point of view by introducing the concept of \emph{minimum $\textsc{d}$-Temporal Out-Spanning Subgraph} ($\textsc{d}$-\TSS) of a temporal graph, as a temporal subgraph that connects the root to any other vertex with walks realizing a given distance $\textsc{d}$ and that minimizes the number of temporal arcs (see Section \ref{sec:TSS} for formal definition)\footnote{Notice that the concepts of spanning \TOB and minimum \TSS coincide, as well as the concepts of spanning $\EA$-\TOB and minimum $\EA$-\TSS. This does not hold for all the other distances.}. Another question then arises:
\begin{enumerate}
    \item[5.] Can we compute a minimum $\textsc{d}$-\TSS of a temporal graph?
\end{enumerate}
In this paper we solve all these questions.

\paragraph{Our contribution.} We first show some characterizations of \TOBs, each of which gives different insights on these objects. Then, for each $\textsc{d}\in \{\text{\LD,\MT,\ST}\}$, we provide a necessary and sufficient condition for the existence of a spanning $\textsc{d}$-\TOB in a temporal graph relying on the concept of optimal substructure (question 1.).
 Moreover, we characterize the vertex set of maximum size that a $\textsc{d}$-\TOB can span, which turns out to be uniquely identified (question 2.). This property is crucial to find efficient polynomial-time algorithms for computing a maximum $\textsc{d}$-\TOB (question 3.). In particular, our algorithms compute a $\textsc{d}$-\TOB whose paths from the root also arrive the earliest possible in every vertex.
 The characterization does not hold for $\textsc{d}\in \{\FT,\MW\}$, and in fact we show that in these cases computing a maximum $\textsc{d}$-\TOB is  an \NP-complete problem (question 3.). We then show that the same results hold for optimal temporal in-branchings (question 4.).
Finally, we prove that for any distance $\textsc{d}$, the problem of finding a minimum $\textsc{d}$-\TSS is again \NP-complete (question 5.).  
 A summary of our results and of the computational time of our algorithms can be found in Table \ref{tab:results}. 
 We stress that any algorithm computing $\textsc{d}(r,v)$ for all vertices $v$ of a temporal graph cannot suffice by itself to find a $\textsc{d}$-\TOB or a minimum $\textsc{d}$-\TSS. Indeed we have seen in Figure (\ref{fig:no_MTLD-tob}) and (\ref{fig:no_ST-tob}) that $\textsc{d}(r,y)$ is well-defined because $y$ is temporally reachable from the root $r$, but no $\textsc{d}$-\TOBs can span $y$. In other words, there are no guarantees that the union of the paths realizing $\textsc{d}$, and computed by the aforementioned algorithms, would form a $\textsc{d}$-\TOB. Also applying the Dijkstra's algorithm on the static expansion of a temporal graph would not solve the problem (see Remark \ref{rem:static_exp}). 
 In addition, for $\textsc{d}\in \{\FT,\MW\}$\, we have the extreme case where computing $\textsc{d}(r,v)$ is polynomial-time, but finding a maximum \textsc{d}-\TOB or a minimum \textsc{d}-\TSS is \NP-complete.
 
\begin{table}[t]
\centering
\caption{Our contribution: summary results. The first row gives the time to compute a \TOB/\TIB/minimum \TSS. The other rows give the time to compute a $\textsc{d}$-\TOB/$\textsc{d}$-\TIB/minimum $\textsc{d}$-\TSS for the corresponding distance $\textsc{d}$. Results marked with $^*$ are presented in \cite{huang2015}.} 
\label{tab:results}
\centering
\begin{tabular}{|c|l|l|l|}
    \hline
$\textsc{d}$ & $\textsc{d}$-\TOB & $\textsc{d}$-\TIB & minimum $\textsc{d}$-\TSS\\ 
    \hline
\emph{none}& $O(m)^*$ & $O(m)$ & equiv. to \TOB\\
\EA & $O(m)^*$ & $O(m\log m)$ & equiv. to $\EA$-\TOB \\
\FT & \NP-complete & \NP-complete& \NP-hard\\
\LD & $O(m\log m)$ & $O(m)$& \NP-hard\\
\MT & $O(m\log n)$ & $O(m\log n)$ & \NP-hard\\
\MW & \NP-complete & \NP-complete & \NP-hard\\
\ST & $O(m\log m)$ & $O(m\log m)$& \NP-hard\\
    \hline
\end{tabular}

\end{table}

\paragraph{Further Related Results.}
We have already mentioned the results of \cite{huang2015}, where the authors also show that the problem of finding minimum weight spanning \TOBs is \NP-hard.
Kuwata et.\ al.\ \cite{Kuwata2009} are interested in the temporal reachability from the root that realizes the earliest arrival time, and they obtain it by making use of Dijsktra's algorithm on the static expansion of the temporal graph. We already observed that that construction does not translate into a \TOB in the original temporal graph. Gunturi et.\ al.\ \cite{gunturi2021} present a polynomial-time algorithm for computing what they call \emph{minimum (weight) spanning tree} in a spatio-temporal network: the difference is that in their model, the weight of the arcs depend on a function that evolves in time but walks are not required to be time-respecting. Different versions of the problem of finding arc-disjoint \TOBs in temporal graphs are investigated in \cite{marino2021,KAMIYAMA2015321}. 
The concept of \TSS is closely related to the one of \emph{spanner}. A spanner of a temporal graph $\mc G=(V,A,\tau)$ is a temporal subgraph of $\mc G$ with vertex set $V$ such that every vertex temporally reaches any other vertex. In contrast, in a \TSS we are only interested in the reachability of all vertices from the root. A \emph{minimum} spanner is a spanner with the least number of temporal arcs possible. Akrida et.\ al.~\cite{AKRIDA2017} proved that computing a minimum spanner is APX-hard. 
It is worth remarking that a very recent preprint~\cite{casteigts2023} introduces new objects that are a relaxation of spanners. Still, they differ from \TOBs and \TSSs for their underlying structure and their reachability properties.

\paragraph{Structure of the paper.}
Section \ref{sec:preliminaries} lists the notation and introduces the concepts used in the paper. In Section \ref{sec:tobTIB} we formally defines the Temporal Out-(In-) Branching and present some preliminary results; in particular we show the equivalence between problems on $\textsc{d}$-\TIBs and on $\textsc{d}$-\TOBs. Section \ref{sec:alg} shows the theoretical results on spanning/maximum $\textsc{d}$-\TOBs; the polynomial-time algorithms for computing a maximum $\textsc{d}$-\TOB are then presented in Subsection \ref{subsec:algMT} when $\textsc{d}=\MT$ and in Subsection \ref{subsec:algLDST} for $\textsc{d}\in \{\LD,\ST \}$. Section \ref{sec:FT-tobs} shows that the related problems for $\textsc{d}\in \{\FT,\MW\}$ are \NP-complete. Finally, in Section \ref{sec:TSS} we formally define a $\textsc{d}$-Temporal Out-Spanning Subgraph and we prove that, for any $\textsc{d}$, the problem of finding a minimum $\textsc{d}$-\TSS is \NP-hard. 

\paragraph{Previous version.}
A preliminary version of this work has been presented at Fundamentals of Computation Theory 2023 \cite{Bubboloni2023}. Compared to that version, here we have added the full proofs of each result, together with some intermediary result and explanatory remarks/examples (namely, Remark \ref{rem:static_exp}, Lemma \ref{claim:subopt}, Figure \ref{fig:no_substitution_FTMW}, Lemma \ref{claim:STtau}, Proposition \ref{prop:gaps_FTMW}, Remark \ref{rem:digraph}).
Also the pseudocode of Algorithm \ref{alg:MTST} has been added.
Moreover, the results on the Minimum Waiting time distance are new as well as all the results of Section \ref{sec:TSS}.

\section{Preliminaries}\label{sec:preliminaries}
We denote by $\mathbb{N}$ the set of positive integers. We set $\mathbb{N}_0=\mathbb{N}\cup\{0\}$, $[n]:=\{x\in \mathbb{N}: x\leq n\}$ and $[n]_0:=\{x\in \mathbb{N}_0: x\leq n\}$, for $n\in \mathbb{Z}$. Note that if $n\in \mathbb{Z}$ and $n<0$, then $[n]=[n]_0=\varnothing.$ We instead have $[0]=\varnothing$ and $[n]_0=\{0\}.$
Given a set $ Q$ and a property $\mc P$, we say that $ Q$ is \emph{minimal} for property $\mc P$ if $ X$ has property $\mc P$, and for all $R\subsetneq Q$, $ R$ does not have property $\mc P$. 
We remind that a digraph is a pair $ D=(V,A)$ where $V $ is the nonempty finite set of vertices and $A\subseteq V\times V$ is the set of arcs. 
Such digraph is called an \emph{out-branching} (respectively \emph{in-branching}) with root $r\in V$ if, for every $v\in V,$ there exists a unique $(r,v)$-\emph{walk} (respectively $(v,r)$-\emph{walk}) in $ D$. 
Note that in a branching, every walk is necessarily a path and that the underlying graph is acyclic.
 We will use the following well-known characterization of out-branchings.
 \begin{lemma}[\cite{gallier2011}, Theorem 4.3]\label{lem:outbranch}
Let $ D=(V,A)$ be a digraph and $r\in V.$ The following facts are equivalent:
\begin{enumerate}
    \item $ D$ is an out-branching with root $r$;
    \item for every $  v\in V\setminus\{r\}$, there exists a $(r,v)$-walk, $d\sub { D}^-(r)=0$ and $ d\sub { D}^-(v)=1$;
    \item for every $ v\in V$, there exists a $(r,v)$-walk and $|A|=|V|-1$.
  \end{enumerate}
\end{lemma}
\noindent A \emph{multidigraph} is formalized by a quadruple $\mc D=(V,A, \tail,\head)$, where $V$ is the set of vertices, $A$ the set of arcs and $\tail, \head:A\to V$ are respectively the \emph{head} and the \emph{tail} function, where we require that $\forall a\!\in\! A$, $\tail(a)\neq \head(a)  $, i.e.\ no selfloops are allowed\footnote{Notice that if $\tail$ and $\head$ are injective, $\mc D$ is a digraph.}.
The \emph{in-neighborhood} and \emph{out-neighborhood} of a vertex $v$ are defined as $N\sub{\mc D}^-(v):=\{u\in V: \exists a\!\in \!A \text{ s.t. }\tail(a)=u,  \head(a)=v \} $ and $N\sub{\mc D}^+(v):=\{u\in V: \exists a\!\in\! A \text{ s.t. }\tail(a)=v,  \head(a)=u  \}$. The \emph{in-degree} and \emph{out-degree} of $v$ are defined respectively as $d\sub{\mc D}^-(v):= |\{a\!\in\! A: \head(a)=v\}|$,  $d\sub{\mc D}^+(v):= |\{a\!\in\! A: \tail(a)=v\}  |.$ Let $u,v\in V$. A $(u,v)$-\emph{walk} of length $k\in \mathbb{N}_0$ in $\mc D$ is an alternating ordered sequence 
$
W=(v_0=u,a_1,v_1,\dots ,v_{k-1},a_k,v_k=v)
$
 of vertices $v_0,v_1,\dots ,v_k\in V$ and arcs $a_1,\dots , a_k\in A$ such that ${\rm t}(a_1)=u $, ${\rm h}(a_k)=v $ and  ${\rm h}(a_i)=v_{i}={\rm t}(a_{i+1})  $ for all $i\in [k-1]$. The set of vertices of $W$ is defined by $V(W):=\{v_0,v_1,\dots ,v_k\}$ 
 and the set of arcs of $W$ by $A(W):=\{a_1,\dots , a_k\}$. Note that $|V(W)|\leq k$ as well as $|A(W)|\leq k.$
 We use the notation $\ell(W)$ for the length $k$ of $W.$ Note that $\ell(W)=0$ if and only if $u=v$; in this case $W$ reduces to the single vertex $u=v$ and it is called a trivial walk.
 We say that $W$ \emph{traverses} a vertex $v$ (an arc $a$) if $v\in V(W)$ ($a\in A(W)$). 
  A \emph{path} is a walk where the vertices are all distinct. If a walk $X$ is a sub-sequence of the walk $W$ is called a \emph{subwalk} of $W$ and we write $X\subseteq W$.
For $h\in[k]_0$ the \emph{$v_h$-prefix} of $W$ is the subwalk of $W$ given by $(v_0,a_1,\dots, a_h ,v_h);$ the \emph{$v_h$-suffix} of $W$ is the subwalk of $W$ given by $(v_h, a_{h+1},\dots,a_k, v_{k}).$ Note that, for a fixed $z\in V(W)$, there are, in general, many $z$-prefixes and many $z$-suffixes of $W$; they are unique for all $z\in V(W) $ if and only if $W$ is a path.
Given a $(u,v)$-walk $W$ and a $(v,s)$-walk $Z$, we denote the walk obtained by their concatenation by $W+Z$.
For $V'\subseteq V,$ the {\it multidigraph induced} by $V'$ in $\mc D$ is the multidigraph $\mc D[V']=(V',A')$, where $A'=\{a\in A: \head(a),\tail(a)\in V' \}$ .

\myparagraph{Temporal Graphs.}
A temporal graph $\mc G$ is a triple $(V,A,\tau)$, where $V $ is the set of vertices, $\tau\in \mathbb{N}$ is the \emph{lifetime}, and $$A\subseteq \{(u,v,s,t)\in V^2\times [\tau]^2 :  u\neq v \text{ and }  s\leq t  \}$$ is the set of \emph{temporal arcs}. We set $m:=|A|$ and $n:=|V|$. Given $a\in A$, we write $a=(\tail (a), \head (a), t_s(a), t_a(a)) $, where $\tail (a)$ and $\head (a)$ are, respectively, the \emph{tail} and \emph{head} vertices of the temporal arc $a$, and $t_s(a)$ and $t_a(a)$ are, respectively, the \emph{starting time}  and the \emph{arrival time} of $a$. 
These functions are easily interpreted: $t_s(a)$ is the time at which it is possible to begin a trip along $a$ from vertex ${\rm t}(a)$ to vertex ${\rm h}(a)$, and $t_a(a)$ is the arrival time of that trip. We also define $el(a):=t_a(a)-t_s(a)$ as the \emph{elapsed} time of the arc $a\in A.$

The temporal graph $\mc G$ has the multidigraph $\mc D\sub{\mc G}=(V,A,\tail, \head)$ as underlying structure. When using concepts like in-neighborhood, out-neighborhood, in-degree and out-degree for a temporal graph $\mc G$, it is intended that we are referring to its underlying multidigraph $\mc D\sub{\mc G}$. Given a temporal graph $\mc G$, we also use the notation $V(\mc G)$ and $A(\mc G)$ to denote, respectively, its set of vertices and temporal arcs.
A temporal graph $\mc G'=(V',A',\tau ')$ is 
a \emph{temporal subgraph} of $\mc G=(V,A,\tau)$ if $V'\subseteq V$, $A'\subseteq A$ and $\tau'\leq \tau$.  
Let $(u,v)\in V^2$; we now introduce the concept of \emph{temporal} $(u,v)$-walk of length $k\in \mathbb{N}_0$. If $k\in\{0,1\}$, every $(u,v)$-walk of length $k$ in the underlying multidigraph is also called a temporal $(u,v)$-walk of length $k$ in $\mathcal{G}$. Let $k\geq 2.$ 
A temporal $(u,v)$-walk of length $k$ in $\mathcal{G}$ is a $(u,v)$-walk $ W=(u,a_1,v_1,\dots ,v_{k-1},a_k,v) $  in the underlying multidigraph such that $t_a(a_i)\leq t_s(a_{i+1}) $ for all $i\in [k-1]$. 
We denote by $\mc W\sub {\mc G}(u,v) $ the set of temporal walks from $u$ to $v$ in $\mc G$. 
If $\mc W\sub {\mc G}(u,v)\neq\varnothing, $ we say that $v$ is \emph{temporally reachable} from $u$. Observe that from every temporal walk it is possible to extract a temporal path with the same end vertices. 

\noindent We now define several interesting functions from the set $\mc W:=\bigcup_{(u,v)\in V^2}\mc W\sub {\mc G}(u,v)$ of the walks of $\mc G$ to $\mathbb{N}_0$. 
Let $W=(u,a_1,v_1,\dots ,v_{k-1},a_k,v)\in \mc W$. If $k\geq 1$, we define the \emph{starting time} of $W$ by $t_s(W):=t_s(a_1)$; the \emph{arrival time} of $W$ by $t_a(W):=t_a(a_k)$; the \emph{duration} of $W$ by $ \dur (W):=t_a(W)-t_s(W)$; the \emph{waiting time} of $W$ by $\wait(W):=\sum_{i=1}^{k-1} [t_s(a_{i+1})-t_a(a_i)] $ if $k\geq 2$ and $\wait(W):=0$ if $k=1$; the \emph{travelling time} of $W$ by $\weight (W):=\sum_{i=1}^k el(a_i)$. If instead $k=0$, i.e.\ $W=(u )$ is a trivial walk, we set $t_s(W)=t_a(W)=\dur (W)=\wait(W)=\weight (W)=0$.
 We next define, through the functions above, a set of crucial functions from $V^2$ to $\mathbb{N}_0\cup\{+\infty\}$. Let first $(u,v)\in V^2$ be such that $\mc W\sub {\mc G}(u,v)\neq \varnothing. $ We define:
\begin{description}
\item[Earliest Arrival time.] $\EA\sub{\mc G}(u,v):=\min\{t_a(W): W \in \mc W\sub {\mc G}(u,v)\}$;
\item[Fastest Time.] $\FT\sub{\mc G}(u,v):=\min\{\dur (W): W \in \mc W\sub {\mc G}(u,v) \}$;
\item[Latest Departure time.] $\LD\sub{\mc G}(u,v):=\max\{t_s(W): W \in \mc W\sub {\mc G}(u,v)\}$ if $u\neq v$, $\LD\sub{\mc G}(u,v):=\tau +1$ if $u=v$;
 \item[Minimum Transfers.] $\MT\sub{\mc G}(u,v):=\min\{\len(W): W \in \mc W\sub {\mc G}(u,v) \}$;
 \item[Minimum Waiting time.] $\MW\sub{\mc G}(u,v):=\min\{\wait(W): W \in \mc W\sub {\mc G}(u,v) \}$. 
  \item[Shortest Travelling time.] $\ST\sub{\mc G}(u,v):=\min\{\weight(W): W \in \mc W\sub {\mc G}(u,v) \}$;
\end{description}
In this case, given $\textsc{d}\!\in \!\{\EA, \FT,\LD, \MT,   \MW,\ST \}$, we say that a temporal $(u,v)$-walk realizes $\textsc{d}\sub{\mc G}(u,v)$ if it attains the minimum (or maximum when \textsc{d}=\LD) that defines the function value $\textsc{d}\sub{\mc G}(u,v)$. If $\mc W\sub {\mc G}(u,v)= \varnothing $, then for every $\textsc{d}\!\in \!\{\EA, \FT,\LD, \MT,   \MW,\ST \}$ we set $\textsc{d}\sub{\mc G}(u,v)=+\infty.$ 

\noindent We will call the functions $\textsc{d}\sub{\mc G}$ (temporal) \emph{distances}, as it is common in the literature~\cite{Calamai2022}. However, notice that  they do not necessarily satisfy the classic requirements for distances, such as the triangle inequality.
When the temporal graph is clear from the context, we usually omit the subscripts.

\section{Temporal branching and preliminary results}\label{sec:tobTIB}




\subsection{Temporal out-branching}\label{sec:temp_out}

In this section, we present the formal notion of temporal out-branching, give some useful characterizations, and define the related optmization problems.

\begin{definition}\label{def:tob}
{\rm A temporal graph $\mc B=(V_{\mc B},A_{\mc B},\tau_{\mc B})$ is called a \emph{temporal out-branching} (\TOB) with root $r\in V_{\mc B}$ if
$A$ is a minimal set of temporal arcs such that, for every $ v\in V_{\mc B}$, there exists a temporal $(r,v)$-walk in $\mc B$.
 If $\mc B$ is a temporal subgraph of a temporal graph  $\mc G=(V,A,\tau)$  we say that $\mc B$ is a \TOB of $\mc G$ rooted in $r$. If in addition $ V_{\mc B}=V$, we say  that $\mc B$ is a \emph{spanning} \TOB of $\mc G$ rooted in $r$. }
\end{definition}
 
\noindent The following lemma provides different characterizations of a \TOB.

\begin{lemma}\label{lem:tob}
Let $\mc B\!=\!(V,A,\tau)$ be a temporal graph and $\mc D$ be the underlying multidigraph of $\mc B$. The following facts are equivalent:
\begin{enumerate}
    \item $ \mc B$ is a {\rm TOB} with root $r$; 
    \item\label{lem:tob_item2} For every $v\in V\setminus \{r\}$, there is a temporal $(r,v)$-walk in $\mc B$, $d\sub { \mc B}^-(r)=0$ and $ d\sub {\mc B}^-(v)=1$;
    \item\label{lem:tob_item3} For every $v\in V$, there is a temporal $(r,v)$-walk in $\mc B$, and $|A|=|V|-1$;
    \item $\mc D$ is a digraph which is an out-branching with root $r$ and, for every $v\in V\setminus \{r\}$, the unique $(r,v)$-walk in $\mc D$ is the unique temporal $(r,v)$-walk in $\mc B$. 
\end{enumerate}
\end{lemma}

\begin{proof}
$\emph{1.} \implies \emph{2.} $
Since the existence of a temporal $(r,v)$-walk in $\mc B$ for all $v\in V$ is guaranteed by definition, we just need to show that $d \sub {\mc B}^-(r)=0 $ and that, for every $v\in V\setminus \{r\}$, $d\sub {\mc B}^-(v)=1 $. 
To that purpose, we first describe the set $A$ of arcs of $\mc B$.
Since from every temporal walk it is possible to extract a temporal path with the same extremes, we have that there exists 
a $(r,v)$-path in $\mc B$ for all $v\in V.$
Fix now one $(r,v)$-path $P_v$ for each $v\in V.$ By the minimality of $A$, we deduce that
\begin{equation}\label{archiT}
A= \bigcup_{v\in V}A(P_v).    
\end{equation} 
As an immediate consequence of \eqref{archiT}, there exists no arc in $A$ entering in $r$ and hence $d\sub{\mc B}^-(r)=0 $.
Now suppose, by contradiction, that there exists $v\in V\setminus \{r\}$ such that $d\sub{\mc B}^-(v)\neq 1 $. If $ d\sub{\mc B}^-(v)=0$, then $v$ is not reachable from $r$, a contradiction. Thus we must have $d\sub{\mc B}^-(v)\geq 2$.
Let $a_1,a_2\in A$ be two different incoming arcs of $v$ with $t_a(a_1)\leq  t_a(a_2).$ We claim that we can delete the temporal arc $a_2$ from $A$ while maintaining the property that every vertex is temporally reachable from $r$, and thus contradicting the minimality of $A$.  Delete $a_2$. By \eqref{archiT}, there exists $v_1\in V$ such that $a_1\in A(P_{v_1}).$ Since in a path there cannot be two different arcs entering the same vertex, we have that $a_2\notin A(P_{v_1}),$ because $a_2\neq a_1.$ 
In particular, $a_2$ is not an arc for the $v$-prefix $X$ of $P_{v_1}.$ 
Let $u\in V$ and consider $P_u.$ If $a_2\notin A(P_u)$, surely $u$ is temporally reachable from $r$ after the removal of $a_2$.
Assume next that  $a_2\in A(P_u)$. Then, since in a path an arc appears at most once, we have that $a_2$ does not appear in the $v$-suffix $S$ of $P_u$.
We consider then the $(r,u)$-walk given by $P=X+S$. Note that $a_2\notin A(P)$ and that $P$ is temporal because $t_a(a_1)\leq  t_a(a_2)$. Hence, again, $u$ is temporally reachable from $r$ after the removal of $a_2$.
 \\
$\emph{2.} \implies \emph{3.}$ Since, by assumption, we have  $d\sub {\mc D}^-(r)=0$ and $ d\sub {\mc D}^-(v)=1$, the fact that $|A|=|V|-1$ follows from $\emph{2.}$ implies $\emph{3.}$ in Lemma \ref{lem:outbranch}. The temporal reachability is trivially true, by assumption.
\\
$\emph{3.} \implies \emph{4.} $  By $\emph{3.}$ implies $\emph{1.}$ in Lemma \ref{lem:outbranch}, we have that $\mc D$ is an out-branching with root $r$. In particular, $\mc D$ is a digraph. Let now $W$ be a temporal $(r,v)$-walk in $\mc B$. Then $W$ is also a $(r,v)$-walk in $\mc D$. By definition of out-branching,  there is a unique $(r,v)$-walk in $\mc D$, so $W$ is the unique temporal $(r,v)$-walk in $\mc B$.
\\
$ \emph{4.} \implies \emph{1.}$
The temporal reachability from the root is guaranteed by hypothesis. Since $\mc D$ is an out-branching, by Lemma \ref{lem:outbranch}, it has $|V|-1$ arcs. Thus, if we delete any arc, then we necessarily disconnect some vertex from the root. Hence $A$ is minimal. 
\end{proof}
\noindent Note that, in particular, Lemma \ref{lem:tob} guarantees that in a \TOB there is a unique temporal walk from the root to any vertex and such a walk is necessarily a temporal path.


\noindent We now want to specialize the concept of \TOB to the various distances. The idea is that we are not only interested in temporally reaching the maximum number of vertices from the root, but we want also to minimize their distance from the root, according to the chosen distance.
\begin{definition}\label{def:Dtob}
{\rm  Let $\mc G=(V,A,\tau)$ be a temporal graph, $r\in V$ and $\textsc{d}\in \{\EA, \FT,\LD, \MT,   \MW,\ST \}$. A \TOB $\mc B=(V\sub {\mc B},A\sub {\mc B},\tau\sub {\mc B})$ of $\mc G$,  rooted in $r$, is called a $\textsc{d}$-\TOB of $\mc G$ rooted in $r$ if, for every $v\in V\sub {\mc B}$, we have  $ \textsc{d}\sub{{\mc B}}(r,v)=\textsc{d}\sub{\mc G}(r,v)$. If in addition $\mc B$ is spanning, we say that $\mc B$ is a \emph{spanning} $\textsc{d}$-\TOB of $\mc G$ rooted in $r$; if instead, in addition,  $|V_\mc B|$ is the largest possible among all the $\textsc{d}$-\TOB of $\mc G$ rooted in $r$, we say that  $\mc B$ is a \emph{maximum} $\textsc{d}$-\TOB of $\mc G$ rooted in $r$.}
\end{definition}
\begin{remark}\label{rem:tob}
 Let  $\mc G=(V,A,\tau)$ be a temporal graph and $r\in V$. Then the following facts hold:
 \begin{itemize}\item[$(i)$] 
Let $\{r\}\subseteq Z\subseteq V$. Then $\mc G$ has a \TOB $\mc B$ rooted in $r$, with vertex set $Z$, if and only if every $v\in Z$ is temporally reachable from $r$ in $ \mc G$. In particular, $\mc G$ has a spanning \TOB $\mc B$ rooted in $r$ if and only if every $v\in V$ is temporally reachable from $r$ in $ \mc G$. 

 \item[$(ii)$] Let $\textsc{d}\in \{\EA, \FT,\LD, \MT,   \MW,\ST \}$. If $\mc B$ is a $\textsc{d}$-\TOB of $\mc G$ rooted in $r$ and with vertex set $Z$, then $\mc B $ is a spanning $\textsc{d}$-\TOB of $\mc G[Z] $ rooted in $r$.
 \end{itemize}
\end{remark}
\begin{proof} $(i)$ If $\mc G$ has a \TOB $\mc B$ rooted in $r$ and with vertex set $Z$, then, by Lemma \ref{lem:tob}, every $v\in Z$ is temporally reachable from $r$ in $ \mc B$ and hence also in  $\mc G$.
Conversely, assume that $Z\subseteq V$ is such that every $v\in Z$ is temporally reachable from $r$ in $ \mc G$. Let $\mc U:=\mc G[Z] $ and let $A_{\mc U}$ be its set of  arcs. If $A_{\mc U}$ is minimal with respect to the reachability in $\mc U$ from the root $r$ then, by  Lemma \ref{lem:tob}, $\mc U$ is a \TOB with vertex set $Z.$ If not, we can delete a finite number of arcs until we reach a minimal set of arcs capable to guarantee the reachability from the root and hence obtain, by Lemma \ref{lem:tob}, a \TOB $\mc B$ with vertex set $Z.$\\
\noindent $(ii)$ Assume that $\mc B$ is a $\textsc{d}$-\TOB of $\mc G$ rooted in $r$ and with vertex set $Z$. Then $\mc B$ is a \TOB and surely $\mc B$ is a temporal subgraph of $\mc U:=\mc G[Z] $. Now, for every $v\in Z$, we have  
$\textsc{d}\sub{{\mc U}}(r,v)\leq \textsc{d}\sub{{\mc B}}(r,v)=\textsc{d}\sub{\mc G}(r,v)\leq \textsc{d}\sub{{\mc U}}(r,v).$
As a consequence, we also have $\textsc{d}\sub{{\mc U}}(r,v)=\textsc{d}\sub{{\mc B}}(r,v).$ This completes the proof.
\end{proof}

\begin{problem}[Maximum {\sc d}-\TOB]\label{prob:tob}
Let $\textsc{d}\in \{\EA, \FT,\LD, \MT,   \MW,\ST \}$ and $\mc G$ be a temporal graph. Find a maximum $\textsc{d}$-\TOB of $\mc G$.
\end{problem}
\noindent Problem \ref{prob:tob} has already been solved for $\textsc{d}=\EA$ in \cite{huang2015}. Their result also implies that a maximum \EA-\TOB spans all the vertices that are temporally reachable from the root, thus the vertex set of a maximum $\EA$-\TOB is uniquely identified. One could be tempted to solve Problem \ref{prob:tob} by simply applying the Dijkstra's algorithm on the static expansion of the temporal graph. This does not produce unfortunately the correct solution. We address the reasons in the following remark.

\begin{figure}
\begin{subfigure}[t]{0.5\textwidth}
\begin{tikzpicture}[scale=0.4]
\SetVertexStyle[FillColor=white]
\SetEdgeStyle[Color=black,LineWidth=1pt]
  \Vertex[y=0,x=0,label=r,color=white!70!black]{1}
  \Vertex[y=0,x=4,label=v]{2}
  \Vertex[y=0,x=8,label=x]{3}
  \Vertex[y=0,x=12,label=y]{4}
  \Edge[Direct,label={(1,1)},bend=40](1)(2)
    \Edge[Direct,label={(2,2)},bend=-40,style=dotted](1)(2)
  \Edge[Direct,label={(1,1)},bend=40](2)(3)
  \Edge[Direct,label={(2,2)},bend=-40,style=dotted](2)(3)
  \Edge[Direct,label={(1,1)}](3)(4)
\end{tikzpicture}
\caption{Temporal graph $\mc G$ that has a spanning $\MT$-\TOB \\ with root $r$ (solid arcs)}\label{fig:Temp_exp1}
\end{subfigure}
\begin{subfigure}[t]{0.5\textwidth}
\begin{tikzpicture}[scale=0.4]
\SetVertexStyle[FillColor=white]
\SetEdgeStyle[Color=black,LineWidth=1pt]
 \Vertex[y=0,x=0,label={[0]},position=left,color=white!70!black]{1}
  \Vertex[y=0,x=5,label={[0]},position=right]{2}
    \Vertex[y=-3,x=0,label={[1]},position=left]{3}
  \Vertex[y=-3,x=5,label={[1]},position=right]{4}
   \Vertex[y=-6,x=0,label={[2]},position=left]{5}
  \Vertex[y=-6,x=5,label={[2]},position=right]{6}
   \Vertex[y=-9,x=0,label={[3]},position=left]{7}
  \Vertex[y=-9,x=5,label={[3]},position=right]{8}

  \Vertex[y=0,x=0,label={(r,1)},color=white!70!black]{1}
  \Vertex[y=0,x=5,label={(r,2)}]{2}
    \Vertex[y=-3,x=0,label={(v,1)}]{3}
  \Vertex[y=-3,x=5,label={(v,2)}]{4}
   \Vertex[y=-6,x=0,label={(x,1)}]{5}
  \Vertex[y=-6,x=5,label={(x,2)}]{6}
   \Vertex[y=-9,x=0,label={(y,1)}]{7}
  \Vertex[y=-9,x=5,label={(y,2)}]{8}

    \Edge[Direct,label={0},position=above](1)(2)
  \Edge[Direct,label={0},position=above,style=dotted](3)(4)
  \Edge[Direct,label={0},position=above,style=dotted](5)(6)
  \Edge[Direct,label={0},position=above](7)(8)
  
  \Edge[Direct,label={1},position=left](1)(3)
  \Edge[Direct,label={1},position=left](3)(5)
  \Edge[Direct,label={1},position=left](5)(7)
\Edge[Direct,label={1},position=left](2)(4)
  \Edge[Direct,label={1},position=left](4)(6)

\end{tikzpicture}
\caption{Static expansion of $\mc G$ and out-branching as output of Dijkstra's algorithm (solid arcs). Numbers in square bracket represent the $\MT$ distance computed.}\label{fig:Temp_exp2}
\end{subfigure}
\caption{Explanatory figure for Remark \ref{rem:static_exp}.} 
\label{fig:Temp_exp}
\end{figure}
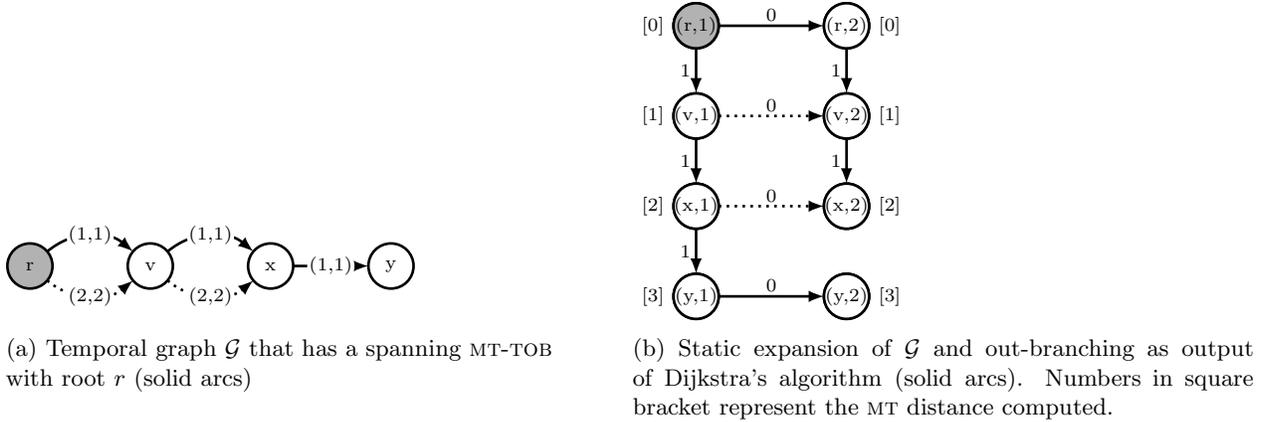

\begin{remark}\label{rem:static_exp}
We remind that the static expansion of a temporal graph $\mathcal{G}=(V,A,\tau)$ is the digraph $SE(\mc G)=(\mc V, \mc A)$ where $\mc V= \{(v,t): v\in V,\, t\in [\tau]  \}$ and $\mc A= \mc M \cup \mc W$ where $ \mc M=\{((u,s),(v,t)): (u,v,s,t)\in A  \}$ and $ \mc W= \{((v,t),(v,t+1)): v\in V, t\in [\tau-1]   \}$, see also \cite{latapy2018,Michail2016}.
The static expansion can be used for computing single source distances $\textsc{d}\sub {\mc G}(r,v)$ by Dijkstra's algorithm, providing each arc in $\mc A$ a suitable weight.
Unfortunately the out-branching that the Dijkstra's algorithm returns on $SE(\mc G)$ does not translate into a $\textsc{d}$-\TOB of $\mc G$. The first problem is that by collapsing back all the vertices $\{ (v,t):t\geq 1  \}$ to $v$, it is not guaranteed that the indegree of $v$ will remain equal to $1$. 
For example, consider the temporal graph of  Figure (\ref{fig:Temp_exp1}) and the out-branching produced by Dijkstra's algorithm for the distance $\textsc{d}=\MT$ on its static expansion in Figure (\ref{fig:Temp_exp2}): if we collapse the vertices of the out-branching, we get as a result the original temporal graph itself. 
Moreover, notice that the vertices $(v,1)$ and $(v,2)$ reach the same distance through the Dijkstra's algorithm (numbers in square brackets), and the same happens to the vertices $(x,1)$ and $(x,2)$; but only the choice of $(v,1)$ and $(x,1)$ would let us achieve a maximum $\MT$-\TOB, while the choice of $(v,2)$ and $(x,2)$ would not. A similar example can be produced for the other distances.  
\end{remark}

\noindent The following concepts will allow us to establish a necessary and sufficient condition for the existence of a spanning $\textsc{d}$-\TOB with root $r$ in a temporal graph in Section \ref{sec:alg}.

\begin{definition}\label{def:prefix-optimal}
Let $\mc G$ be a temporal graph and $W$ be a temporal $(u,v)$-walk in $\mc G$. For every $\textsc{d}\in \{\EA,\FT,\LD,\MT,\MW,\ST \}$ we say that: 
 \begin{itemize}
     \item $W$ is $\textsc{d}$-\emph{prefix-optimal} if, for every $ x\in V(W)$, any $x$-prefix of $W$ realizes $\textsc{d}\sub {\mc G}(u,x)$;
     \item $W$ is $\EA\textsc{d}$-\emph{prefix-optimal} if it is $\textsc{d}$-\emph{prefix-optimal} and, for every $x\in\ V(W)$, any $x$-prefix of $W$ realizes $\EA\textsc{d}\sub {\mc G}(u,x)$.
 \end{itemize}
\end{definition}
\subsection{Temporal in-branching}\label{sec:TIB}

In this section, we present definitions of temporal in-branchings and prove that the related problems are computationally equivalent to \TOBs. 

\begin{definition}
{\rm A temporal graph $\mc B=(V\sub {\mc B},A\sub {\mc B},\tau\sub {\mc B})$ is called a \emph{temporal in-branching} (\TIB) with root $r$ if $A\sub {\mc B}$ is a minimal set of temporal arcs such that for all $ v\in V$, there exists a temporal $(v,r)$-walk in $\mc B$.
Given $\mc G=(V,A,\tau)$ a temporal graph, $r\in V$, and $\mc B=(V\sub {\mc B},A\sub {\mc B},\tau\sub {\mc B})$ a temporal subgraph of $\mc G$ that is a \TIB rooted in $r$, we say that $\mc B$ is \emph{spanning} if $V\sub {\mc B}=V$ and \emph{maximum} if $|V\sub {\mc B}|$ is the largest possible.
Given $\textsc{d}\in \{\EA,\LD,\MT,\FT,\ST,\MW \}$ and $\mc B$ a \TIB with root $r$ of $\mc G$, we say that $\mc B$ is a $\textsc{d}$-\TIB of $\mc  G$ if for every $v\in V\sub  { \mc B}$, $ \textsc{d}\sub{ { \mc B}}(v,r)=\textsc{d}\sub{\mc G}(v,r)$. If in addition $V\sub {\mc B}=V$, then $\mc B $ is a  \emph{spanning} $\textsc{d}$-\TIB, and if $|V\sub {\mc B}|$ is the largest possible among the $\textsc{d}$-\TIB rooted in $r$, then $\mc B$ is a \emph{maximum} $\textsc{d}$-\TIB.}
\end{definition}

\begin{problem}[Maximum {\sc d}-\TIB]\label{prob:tib}
Let $\textsc{d}\in \{\EA, \FT,\LD, \MT,   \MW,\ST \}$ and $\mc G$ be a temporal graph. Find a maximum $\textsc{d}$-\TIB of $\mc G$.
\end{problem}
\noindent The next proposition shows that finding maximum \TIBs can be reduced to finding maximum \TOBs in an auxiliary temporal graph; we first need to define the transformation $\cir$ of a temporal graph that reverses the order of the timesteps as well as the direction of the arcs. A similar transformation has been used e.g. in~\cite{Calamai2022}.
\begin{definition}\label{def:reverse}
Given a temporal graph $\mc G=(V,A,\tau)$, we define the \emph{reverse} of $\mc G$ as the temporal graph $\mathcal{G}^{\circlearrowleft}=(V,A^{\cir},\tau)$ where 
 $A^{\cir} =\{ (\head(a),\tail(a),\tau-t_a(a)+1,\tau-t_s(a)+1  ):a\in A\}:=\{a^{\cir}:a\in A \}$. 
\end{definition}


\begin{proposition}\label{prop:equiv_tobtib}
Given a temporal graph $\mc G$, it holds that:
\begin{enumerate}
    \item $\mc B$ is a maximum \EA-\TIB of ${\mc G}$ if and only if $\mc B^{\cir}$ is a maximum \LD-\TOB of $\mathcal{G}^{\circlearrowleft}$;
    \item $\mc B$ is a maximum \LD-\TIB of ${\mc G}$ if and only if $\mc B^{\cir}$ is a maximum \EA-\TOB of $\mathcal{G}^{\circlearrowleft}$;
    \item For each $\textsc{d}\!\in\! \{\FT,\MT,\MW,\ST \} $, $\mc B$ is a maximum $\textsc{d}$-\TIB of ${\mc G}$ if and only if $\mc B^{\cir}$ is a maximum $\textsc{d}$-\TOB of $\mathcal{G}^{\circlearrowleft}$.
\end{enumerate}  
\end{proposition}

\begin{proof}
 Observe that $ (\mc G^{\cir})^{\cir}=\mc G $ and that $W=(u,a_1,v_2, a_2, \dots ,v_k,a_{k}, v)$ is a temporal $(u,v)$-walk in $\mc G$ if and only if $W^{\cir}=(v,a^{\cir}_{k},v_k, \dots , a^{\cir}_2,v_2,a^{\cir}_1,u)$ is a temporal $(v,u)$-walk in $ \mc G^{\cir}$. Then note that for any walk $W$ in $\mc G$ it holds that $(W^{\cir})^{\cir}=W$, $t_s(W^{\cir})= \tau- t_a(W) +1 $, $t_a(W^{\cir})= \tau- t_s(W) +1 $, $ \len(W)=\len(W^{\cir})$, 
 $\dur(W)=\dur(W^{\cir})$, $\wait(W)=\wait(W^{\cir}) $, and , $\weight(W)=\weight (W^\cir) $. It is also easy to see that $\mc B$ is a \TIB with root $r$ if and only if $\mc B^{\cir}$ is a  \TOB with root $r$. 
Let now $W$ and $W'$ be two $(v,r)$-walks in $\mc G$.
We claim that $W$ realizes $\EA(v,r)$ in $\mc G$ if and only if $W^{\cir}$ realizes $\LD(r,v)$ in $\mc G^{\cir}$. Indeed, $t_a(W)\leq t_a(W')$ if and only if $\tau- t_s(W^{\cir})+1\leq \tau- t_s(W'^{\cir})+1$ if and only if $t_s(W^{\cir})\geq  t_s(W'^{\cir})$.
Similarly, $W$ realizes $\LD(v,r)$ in $\mc G$ if and only if $W^{\cir}$ realizes $\EA(r,v)$ in $\mc G^{\cir}$. Indeed,  $t_s(W)\geq t_s(W')$ if and only if $\tau- t_a(W^{\cir})+1\geq \tau- t_a(W'^{\cir})+1$ if and only if $t_a(W^{\cir})\leq  t_a(W'^{\cir})$. 
We now prove that, for $\textsc{d}\in \{\FT,\MT,\MW,\ST \}$, $W$ realizes $\textsc{d}(v,r)$ in $\mc G$ if and only if $W^{\cir}$ realizes $\textsc{d}(r,v)$ in $\mc G^{\cir}$. In fact it holds that $\dur(W)\leq \dur(W')$ if and only if $\dur(W^{\cir})\leq  \dur(W'^{\cir})$, $\len(W)\leq \len(W')$ if and only if $\len (W^{\cir})\leq  \len (W'^{\cir})$, $\wait(W)\leq \wait(W')$ if and only if $\wait(W^{\cir})\leq  \wait(W'^{\cir})$, and $\weight(W)\leq \weight(W')$ if and only if $\weight (W^{\cir})\leq  \weight (W'^{\cir})$. This concludes the proof.
\end{proof}

\section{Computing maximum {\sc d}-temporal out-branchings for latest departure time, minimal transfers and shortest time distances}\label{sec:alg}



\noindent Before stating necessary and sufficient conditions for the existence of a spanning $\textsc{d}$-\TOB in a temporal graph, we need the following lemma. It intuitively says that a $v$-prefix in a $\textsc{d}$-prefix-optimal walk can be replaced by another $\textsc{d}$-prefix-optimal walk that arrives in $v$ earlier.
\begin{lemma}\label{claim:subopt} Let $\mc G=(V,A,\tau)$ be a temporal graph, $r,u\in V$ and $\textsc{d}\in \{\MT,\ST,\LD  \} $. 
Let $W$ be a $\textsc{d}$-prefix-optimal temporal $(r,u)$-walk in $\mc G$ and $v\in V(W)$. Let $ S$ be a $v$-suffix of $W$ and let $W_v$ be a $\textsc{d}$-prefix-optimal temporal $(r,v)$-walk in $\mc G$. If $ t_a(W_v)\leq t_s(S)$, then $W_v + S$ is a $\textsc{d}$-prefix-optimal temporal $(r,u)$-walk in $\mc G$.
\end{lemma}
\begin{proof}
Since $ t_a(W_v)\leq t_s(S)$ by hypothesis, then $\Bar{W}:=W_v + S$ is a temporal $(r,u)$-walk. Let $x\in V(\Bar{W})$ and $X$ be a $x$-prefix of $\Bar{W}$. If $\textsc{d}=\LD$, since $W$ and $W_v$ are $\LD$-prefix-optimal, it holds that $ t_s(W)=\LD(r,v)=t_s(W_v) $. Since $\LD(r,v)=\LD(r,x) $ and $t_s(X)=t_s(\Bar{W})=t_s(W_v)$, we obtain that $ t_s(X)= \LD(r,x)$.
Consider now $\textsc{d}\in \{\MT,\ST \}$ and let  $P$ be the $v$-prefix of $W$ such that $W=P+S$.
If $X\subseteq W_v$, then by hypothesis $X$ realizes $ \textsc{d}(r,x)$. If $X\not\subseteq W_v$, there exists $S'\subseteq S$ such that  $X=W_v+S'.$ 
Assume first $\textsc{d}=\MT$. Then $ \ell(X)=\ell(W_v)+\ell(S')=\MT(r,v)+\ell(S')$ and, by hypothesis, $\MT(r,x)=\ell(P+S')=\ell(P)+\ell(S')= \MT(r,v)+\ell(S')$. Consequently $ \ell(X)=\MT(r,x)$.
Assume next that $\textsc{d}=\ST$. Then we have $\weight(X)=\weight(W_v)+\weight(S')=\ST(r,v)+\weight(S')$ 
and, by hypothesis, $\ST(r,x)=\weight(P+S')=\weight(P)+\weight(S')= \ST(r,v)+\weight(S')$. This implies that $ \weight(X)=\ST(r,x)$.  
\end{proof}

\noindent We emphasize that Lemma \ref{claim:subopt} does not hold for $\textsc{d}\in \{\FT, \MW \} $. Indeed consider the temporal graph in Figure \ref{fig:no_substitution_FTMW}. We have that $ \FT(r,v)=5$ and $ \MW(r,v)=1$, which is realized by the solid $(r,v)$-path. In particular, this path is $\FT$-prefix-optimal and $\MW$-prefix-optimal. Consider now the dashed $(r,x_3)$-path and call it $W$: it is both $\FT$-prefix-optimal and $\MW$-prefix-optimal. Let $a$ be the temporal arc $(x_3,v,6,7)$. Notice that $W'=W+(x_3,a,v)$ 
is a temporal $(r,v)$-path but does not realize neither $\FT(r,v) $ nor $\MW(r,v) $, as $\dur(W')= 6$ and $\wait(W')=2$.

\begin{figure}
    \centering
    \begin{tikzpicture}[scale=0.4]
\SetVertexStyle[FillColor=white]
\SetEdgeStyle[Color=black,LineWidth=1pt]
  \Vertex[y=0,x=0,label=r,color=white!70!black]{1}
 
  \Vertex[y=0,x=3,label=$x_1$]{2}
  \Vertex[y=0,x=6,label=$x_2$]{3}
  \Vertex[y=0,x=9,label=$x_3$]{4}
  \Vertex[y=0,x=12,label=v]{5}
 \Vertex[y=-2,x=3, label=$y_1$]{6}
  \Vertex[y=-2,x=6, label=$y_2$]{7}
  
  \Edge[Direct,label={(2,3)},position=above](1)(2)
  \Edge[Direct,label={(3,5)},position=above](2)(3)
  \Edge[Direct,label={(5,5)},position=above](3)(4)
  \Edge[Direct,label={(6,7)},position=above](4)(5)

  \Edge[Direct,label={(1,2)},position=below left,style=dashed](1)(6)
  \Edge[Direct,label={(2,4)},position=below,style=dashed](6)(7)
  \Edge[Direct,label={(4,4)},position=below right,style=dashed](7)(4)
  \Edge[Direct,label={(6,7)},position=above,style=dashed](4)(5)
\end{tikzpicture}
    \caption{Example showing that Lemma $\ref{claim:subopt}$ does not hold for $\textsc{d}\in \{\FT, \MW \} $.}
    \label{fig:no_substitution_FTMW}
\end{figure}
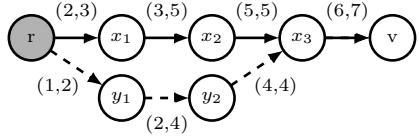

\begin{theorem}
\label{thm:tob}
Let $\mc G=(V,A,\tau)$ be a temporal graph, $r\in V$ and $\textsc{d}\in \{\LD,\MT,\ST  \}$. 
Then $\mc G$ has a spanning $\textsc{d}$-\TOB with root $r$ if and only if there exists a $\textsc{d}$-prefix-optimal temporal $(r,v)$-path in $\mc G$ for all $v\in V$.
\end{theorem}
\begin{proof}
 $\implies$ By the definition of $\textsc{d}$-\TOB with root $r$ and the uniqueness of temporal walks from the root to any other vertex in a \TOB (see Lemma~\ref{lem:tob}), it follows that all walks in a $\textsc{d}$-\TOB are paths and are $\textsc{d}$-prefix-optimal.\\ 
$\impliedby$
For every $v\in V$, let $W_v$ be a temporal $\textsc{d}$-prefix-optimal $(r,v)$-path in $\mc G$. Let $A':=\bigcup_{v\in V} A(W_v).$ 
 For $B\subseteq A',$ denote by $\mc B[B]$ the temporal subgraph of $\mc G$ having vertex set $V$ and temporal arc set $B$, and consider the property.\\
$\mc P\sub{{\textsc{d}, B}}$: for all $v\in V$, there exists in $\mc B[B]$  a temporal $\textsc{d}$-prefix-optimal $(r,v)$-walk.\\
 \noindent Note that $\mc P\sub{{\textsc{d}, A'}}$ is satisfied. Thus it is possible to consider the minimal subsets $B$ of arcs in $A'$ satisfying $\mc P\sub{{\textsc{d}, B}}$. 
 Let $A\sub {\mc B}\subseteq A'$ be one of such minimal sets and let $\mc B:=\mc B[A\sub {\mc B}].$
   We show that $\mc B$ is a $\textsc{d}$-\TOB for $\mc G$. Clearly, by the construction of $\mc B$, it is enough to show
    that $\mc B$ is a \TOB. Furthermore, in view of Lemma \ref{lem:tob}, it suffices to show that $d^-\sub {\mc B}(r)=0$ and that for all $v\in V\setminus \{ r\}$, $d^-\sub {\mc B}(v)=1$, since the temporal reachability from vertex $r$ to any other vertex is already guaranteed by property $\mc P\sub{{\textsc{d}, A\sub {\mc B}}}$. 
    Since $A'$ does not contain arcs entering in $r$, this holds also for $A\sub {\mc B}$ and hence we have that $d^-\sub {\mc B}(r)= 0$.
  Suppose now, by contradiction, that there exists $v\in V\setminus \{r \}$ such that $d^-\sub {\mc B}(v)\neq 1$. Since $d^-\sub {\mc B}(v)=0$ implies that 
 $v$ is not reachable from $r$, we necessarily have $d^-\sub {\mc B}(v)\geq 2$. Let $a_1, a_2\in A\sub {\mc B} $ be two different incoming temporal arcs of $v$  with $t_a(a_1)\leq  t_a(a_2).$ 
We claim that $\mc P\sub{{\textsc{d}, A\sub{{\mc B}}\setminus\{a_2\}}}$ is satisfied, and thus the minimality of $A\sub {\mc B}$ is contradicted. Indeed, by definition of $A\sub {\mc B}$, there exists $v_1\in V$ and a temporal $(r,v_1)$-path $W_{v_1}$
such that $a_1\in A(W_{v_1}).$ Since in a path two distinct arcs entering in the same vertex do not appear, we have that $a_2\notin A(W_{v_1}).$ In particular, $a_2$ is not an arc for the $v$-prefix $X$ of $W_{v_1}.$ Let $u\in V$ and consider $W_u$ a temporal $(r,u)$-path in $\mc B$.  Assume that  $a_2\in W_u$. Then, since in a path an arc appears at most once, we have that $a_2$ does not appear in the $v$-suffix $S$ of $W_u$.
We consider then the $(r,u)$-walk given by $\Bar W=X + S$. Then $a_2\notin A(\Bar W)$ and we have that
$t_a(X)=t_a(a_1)\leq  t_a(a_2)\leq t_s(S).$
As a consequence, by Lemma \ref{claim:subopt}, $\Bar W$ is a $\textsc{d}$-prefix-optimal walk in $\mc B[A\sub {\mc B}\setminus\{a_2\}]$.
\end{proof}
\noindent Theorem \ref{thm:tob} does not hold for $\textsc{d}\in \{\FT, \MW \} $. Indeed the temporal graph in Figure (\ref{fig:no_FT-tob}) has a $\textsc{d}$-prefix-optimal path from $r$ to any other vertex, but does not admit a spanning \textsc{d}-\TOB as previously observed. We are now ready to characterize the vertex set of a maximum $\textsc{d}$-\TOB.
\begin{corollary}
\label{cor:maxDtob}
Let $\mc G=(V,A,\tau)$ be a temporal graph, $r\in V$ and $\textsc{d}\in \{\LD,\MT,\ST  \}$. Then a maximum $\textsc{d}$-\TOB $\mc B$ with root $r$ of $\mc G$ has vertex set:
\begin{equation}
 V\sub {\mc B}=\{v\in V: \text{there exists a $\textsc{d}$-prefix-optimal $(r,v)$-path in $\mc G$ } \}.  \label{eq:tobmax1}
\end{equation}
\end{corollary}
\begin{proof}
 Consider $\mc G[V\sub {\mc B} ]$. Let $v\in V\sub {\mc B}$ and $W$ a $\textsc{d}$-prefix-optimal $(r,v)$-temporal walk in $\mc G$. 
 By definition of $\textsc{d}$-prefix-optimal walk, for every $u\in V(W)$, it holds that $u\in V\sub {\mc B}$, which implies that $W$ is also a $\textsc{d}$-prefix-optimal $(r,v)$-temporal walk in $ \mc G[V\sub {\mc B} ]$. Hence, by Theorem~\ref{thm:tob}, $ \mc G[V\sub {\mc B} ]$ has a spanning $\textsc{d}$-\TOB $\mc B$, which is also a $\textsc{d}$-\TOB of $\mc G$. We now show that $\mc B$ is maximum. By Remark \ref{rem:tob} it suffices to prove that if $V'\subseteq V$ is such that $V'\setminus V\sub {\mc B}\neq \emptyset $, then $\mc G[V']$ does not admit a spanning $\textsc{d}$-\TOB with root $r$. Let $u\in  V'\setminus V\sub {\mc B}$. By hypothesis there does not exist a $\textsc{d}$-prefix-optimal temporal $(r,u)$-walk in $\mc G$, hence there does not exist one in $\mc G[V']$. By Theorem \ref{thm:tob}, $\mc G[V']$ does not admit a spanning $\textsc{d}$-\TOB. 
\end{proof}
\noindent The above corollary shows that for every $\textsc{d}\in \{\LD,\MT,\ST \}$, the vertex set of a maximum $\textsc{d}$-\TOB of a temporal graph is uniquely determined. We now show that for $\textsc{d}=\ST $ and $\tau\leq 2$, a maximum $\ST$-\TOB is always spanning, as long as each vertex is temporally reachable from the root. That is no longer true for $\tau\geq 3$, as already highlighted in Figure (\ref{fig:no_ST-tob}).

\begin{lemma}\label{claim:STtau}
 If $\tau\leq 2$, then a temporal graph $\mc G=(V,A,\tau)$ has a spanning $\ST$-\TOB with root $r$ if and only if each vertex is temporally reachable from $r$.   
\end{lemma}

\begin{proof}
If $\tau\!=\!1$ the temporal graph reduces to a static graph, so the problem reduces to find a spanning out-branching of a static graph. If $\tau\!=\!2$, notice that every temporal label $(t_s(a), t_a(a))$ can assume only three values, namely $\{(1,1), (1,2), (2,2) \}$. This implies that every temporal walk $W$ in $\mc G$ is such that either $\weight(W)=0$ or $\weight(W)=1$. 
We now want to prove that if $v$ is temporally reachable from $r$, then there exists an $\ST$-prefix-optimal $(r,v)$-walk in $\mc G$. Let $W$ be an $(r,v)$-walk in $\mc G$. If $\weight(W)=0$, then $W$ is necessarily $\ST$-prefix-optimal. Suppose now that $\weight(W)=1$ and that $W$ is not $\ST$-prefix-optimal. Let $u$ be the first vertex of $W$, starting from $v$, for which the $u$-prefix $X$ of $W$ does not realize $\ST(r,u)$. If $\weight(X)=0$, then we also have $\ST(r,u)=0,$ against the fact that $X$ does not realize $\ST(r,u)$. As a consequence, we have $\weight(X)=1$ and $\ST(r,u)=0 $. This implies that $t_a(X)=2$ and that there must exist an $(r,u)$-walk $W_u$ in $\mc G$ such that $\weight(W_u)=0$; in particular $W_u$ is $\ST$-prefix-optimal.  Let $S$ be a $u$-suffix of $W$ such that $W=X+S$. Notice that $t_a(W_u)\leq \tau=2=t_a(X)\leq t_s(S)\leq 2 $, so that $t_a(X)= t_s(S)$. Then all the temporal labels of the arcs in  $S$ are equal to $(2,2)$. As a consequence $W_u+S $ is an $\ST$-prefix-optimal $(r,v)$-walk in $\mc G$.  
\end{proof}

\noindent The next sections present algorithms for finding $\textsc{d}$-\TOBs of a given temporal graph in polynomial time, when $\textsc{d}\!\in\! \{\LD,\MT,\ST  \}$. In particular, we show that in such cases we can constrain ourselves to the earliest arrival paths that realize the distances. For this, we define:

\begin{definition}
{\rm Given a temporal graph $\mc G=(V,A,\tau)$ rooted in $r\in V$, for any $(u,v)\in V^2$ and $\textsc{d}\in \{\EA, \FT,\LD, \MT,   \MW,\ST \}$, we define $\EA\textsc{d}\sub{\mc G}(u,v):=\min\{t_a(W): W \text{ realizes } \textsc{d}\sub{\mc G}(u,v)\}$.
A $\textsc{d}$-\TOB $\mc B\!=\!(V\sub {\mc B},A\sub {\mc B},\tau\sub {\mc B})$ with root $r$ of $\mc G$ is called an $\EA\textsc{d}$-\TOB if, for every $v\in V\sub {\mc B}$, we have that $ \EA\textsc{d}\sub{{\mc B}}(r,v)=\EA\textsc{d}\sub{\mc G}(r,v)$. 
$\mc B$ is called \emph{spanning} if $V\sub {\mc B}=V$ and \emph{maximum} if $|V\sub {\mc B}|$ is the largest possible among the $\EA\textsc{d}$-\TOB rooted in $r.$}
\end{definition}

\noindent The proposed algorithms will always return an $\EA\textsc{d}$-\TOB. This implies that for $\textsc{d}\in \{\MT,\ST,\LD \}$, the existence of a $\textsc{d}$-prefix-optimal $(r,v)$-path in $\mc G$ is equivalent to the existence of an $\EA\textsc{d}$-prefix-optimal $(r,v)$-path in $\mc G$. For $\textsc{d}= \FT$ this is no longer true: indeed consider Figure (\ref{fig:no_FT-tob}). The only $\FT$-prefix-optimal $(r,y)$-path is $W=(r,(r,v,2),v,(v,y,2),y)$, but it is not $\EA\FT$-prefix-optimal: in fact, $\EA\FT(r,v)=1 $ since the path $(r,(r,v,1),v)$ realizes $\FT(r,v)$ and arrives in $v$ at time $1$, while $W$ arrives in $v$ at time $2$. The same reasoning holds for $\textsc{d}= \MW$. This difference will be crucial for showing that computing a maximum \textsc{d}-\TOB for $\textsc{d}\in \{\FT,\MW \} $ is an \NP-complete problem (Section \ref{sec:FT-tobs}).


\subsection{Algorithm for minimal transfer distance}\label{subsec:algMT}
Algorithm \ref{alg:MTST} computes a maximum $\MT$-\TOB with root $r$ of a given temporal graph. First observe that, given an $\MT$-prefix-optimal temporal $(r,v)$-walk $W=(r=v_0,a_1,v_1,\ldots ,a_k,v_k=v)$, we have that $\MT(r,v_i)\!=\!\MT(r,v_{i+1})-1<\MT(r,v_{i+1})$ for all $i\in [k-1]$, in particular the sequence of distances in any \MT-prefix-optimal walk is strictly increasing.
\begin{algorithm}[t]
\small
\caption{Computing a maximum $\MT$-\TOB of a temporal graph.}\label{alg:D1-tob}\label{alg:MTST}
\KwIn{A temporal graph $ \mathcal{G}=(V,A,\tau)$, and a vertex $r\in V$.}
\KwOut{A maximum $\MT$-\TOB $\mc B = (V\sub{\mc B},A\sub{\mc B},\tau\sub{\mc B})$ of $\mc G$ with root $r$.}
$\mc {EAMT}(r)\gets 0; \forall v\in V\setminus \{r\}, \mc {EAMT}(v)\gets +\infty $\;
$d(r)\gets 0$; $\forall v\in V\setminus \{r\},d(v)\gets \MT\sub{\mc G}(r,v)$\; 
$V\sub{\mc B}\gets \{r\}$; $A\sub{\mc B}\gets \emptyset$; $\tau\sub{\mc B}\gets 0$; $h\gets \max\{d(v):v\in V, d(v)<+\infty \}$\;
\For{$i=1,\dots, h$}
{
    \For{each $v\in V$ such that $d(v)=i$}
    {$S\gets \{(u',v,s',t')\in A: s'\geq  \mc {EAMT}(u') ,\, d(u')=i-1\}$\;
        \If{$S \neq \emptyset$}
        {
        $a \gets $ choose $ (u,v,s,t)\in \argmin_{(u',v,s',t')\in S} t' $\;
        $\mc {EAMT}(v)\gets t$,  $V\sub{\mc B}\gets V\sub{\mc B}\cup  \{v\}$, $A\sub{\mc B}\gets A\sub{\mc B}\cup  \{a\}$, $\tau\sub{\mc B}\gets \max\{\tau\sub{\mc B},t\}$\;
        }
    }
}
\end{algorithm}
The main idea of the algorithm is then to compute a priori the \MT-distances of all vertices from the root (line 2), and then build the $\MT$-\TOB guided by these computed distances, using their strict monotonicity property. 
More specifically, given $h=\max\{\MT(r,v): v\in V \}$, the algorithm grows an $\MT$-\TOB starting from the root and adding, at step $i\!\in\! [h]$, all the vertices at distance $i$. During this process, when adding some vertex $v$, we choose, among its neighbors at distance $i-1$, which one can be the parent of $v$. 
To choose the right parent, we look at the incoming temporal arcs having tail in vertices at distance $i-1$ and we consider only the arcs $a'=(u',v,s',t')$ such that, if $W_{u'}$ is the unique temporal $(r,u')$-path in the $\MT$-\TOB built so far, then $s'\geq t_a(W_{u'})$, i.e.\ the new arc can be concatenated with $W_{u'}$ to obtain a temporal $(r,v)$-path. Among the arcs fulfilling these constraints, we choose $a'$ minimizing $t'$, the arrival time in $v$; such arc $a'$ exists if and only if there exists an $\MT$-prefix-optimal $(r,v)$-path in $\mc G$. 
We prove that such choice of $a'$ ensures that we are actually representing in the \TOB a temporal $(r,v)$-path that realizes the distance $\MT(r,v)$ and has the earliest arrival time among the walks realizing such distance, i.e.\ we are computing an $\EA\MT$-\TOB.
The algorithm takes $O(m\log n)$ time to compute all the initial $\MT$ distances (Table \ref{table:shortest_path}), while 
the remaining part of the algorithm takes $O(m)$ time as it requires only one scan of each temporal arc.

\begin{theorem}\label{thm:algMTST}
Algorithm $\ref{alg:MTST}$ returns a maximum $\MT$-\TOB of a temporal graph, for a chosen root, in $O(m\log n)$ time. Additionally, the output is an $\EA\MT$-\TOB.
\end{theorem}

\begin{proof}
Let $\mc G=(V,A,\tau)$ be the temporal graph input of the algorithm and $r\in V$, $d(v)=\MT\sub{\mc G}(r,v)$ for all $v\in V$, $h=\max \{d(v):v\in V , d(v)<+\infty\}$ and $V'=\{v\in V: v \text{ is temporally reachable from $r$} \}$. For $i\in [h]_0$ let $D_i=\{v\in V: d(v)=i  \}$ and note that $\{D_i:i\in  [h]_0\}$ is a partition of $V'$ with $D_0=\{r\}$. Since no confusion is possible, from now on we will avoid writing the subscripts $\mc G$. 
We prove the following loop invariant:
\begin{claim}\label{claim:MTSTalg}
At the end of the $i$-th iteration of the \textbf{for} loop in lines 4-12, \[  V\sub{\mc B}=\{v\in V': \exists \text{ an $\MT$-prefix-optimal temporal } (r,v)\text{-walk in } \mc G \text{ and } d(v)\leq i  \},\] $\mc{EAMT}(v)=\EA\MT(r,v)  $ for all $v\in V\sub{\mc B}$, and
 $\mc B$ is an $\EA\MT$-\TOB with root $r$ of $\mc G$.
\end{claim}
\noindent The above claim implies that the final output $\mc B$ of the algorithm is an $\EA\MT$-\TOB with root $r$ of $\mc G$, which is in particular an $\MT$-\TOB. Moreover, $V\sub {\mc B}$ consists of all the vertices in $\mc G$ for which there exists an $\MT$-prefix-optimal temporal walk from the root. Thus $\mc B$ is a maximum $\MT$-\TOB by Corollary \ref{cor:maxDtob}. We are left to prove the claim. 
$\mc B$ is initialized as the temporal graph made of the sole vertex $r$, so the loop invariant is trivially true. 
Suppose now that the loop invariant is true up to a certain $i$-th iteration. We now prove that it holds for the $(i+1)$-th iteration.
Let $v\in V$ be such that $d(v)=i+1$. We first prove that if there exists an $\MT$-prefix-optimal temporal $(r,v)$-walk in $\mc G$, say $W$, then the set $S$ in line 6 is non-empty. We can always choose $W$ such that it arrives the earliest, 
that is $t_a(W)=\EA\MT(r,v)$.
Let $\ba a=(\ba u,v, \ba s, \ba t)\in A$ be the last temporal arc of $W$. 
It holds that $d(\ba u)=i$, so by inductive hypothesis we have that $\ba u\in V\sub{\mc B}$ and $\mc {EAMT}(\ba u)=\EA\MT(r,\ba u)$  at the end of the $i$-iteration. 
Since $W$ is $\MT$-prefix-optimal, we have that $\ba s\geq \EA\MT(r,\ba u)=\mc{EAMT}(\ba u)$. 
Therefore $\ba a$ fulfils the conditions to belong to $S$, so $S$ is non-empty. Notice also that since $ \ba a\in S$ and $\ba t=\EA\MT(r,v)$ then
\begin{equation}\label{eq:t'}
    \min_{(u',v,s',t')\in S}t'= \ba t=\EA\MT(r, v)\, .
\end{equation}
We now prove that $\mc B$ is an $\EA\MT$-\TOB with root $r$. We have just showed that if there exists an $\MT$-prefix-optimal temporal $(r,v)$-walk in $\mc G$, then $S$ in line 6 is non-empty. This implies that in line 8 we choose an arc $a=(u,v,s,t)\in S$ that minimizes the arrival time, and this arc is added to $A\sub {\mc B}$, while $v$ is added to $V\sub {\mc B}$ and $\mc {EAMT}(v)$ is set to $t$. Since $D_0,\ldots,D_h$ is a partition of $V'$, no other incoming arc to $v$ is added in the algorithm, and therefore $v$ has in-degree equal to $1$ in $\mc B$. Moreover $s\geq \mc {EAMT}(u)$ since $a\in S$, so if $W_u$ is the unique temporal $(r,u)$-path in $\mc B$ (it exists by inductive hypothesis), then $W_v=W_u +(u,a,v)$ is a temporal $(r,v)$-path in $\mc B$. Hence $\mc B$ is a \TOB with root $r$. 
It remains to show that $W_v$ realizes $\EA\MT(r,v)$. By the inductive hypothesis we have that $W_u$ is $\EA\MT$-prefix-optimal. Therefore $\len(W_v)= \len(W_u)+1=d(u)+1=i+1=d(v)$. 
Moreover, by equation (\ref{eq:t'}) and since $a\in S$, we have that $\mc{EAMT}(v)= t_a(W_v)=   t=\ba t= \EA\MT(r,v) $. This concludes the proof of claim.

Regarding the computational complexity of the algorithm, by Table \ref{table:shortest_path} the initial computation of all distances requires $O(m\log n)$; 
the remaining part of the algorithm takes $O(m)$ as it requires only one scan of each temporal arc. Therefore the overall complexity is $O(m\log n)$. 
\end{proof}

\subsection{Algorithm for latest departure time and shortest time distances}\label{subsec:algLDST}
Algorithm \ref{alg:LDST-tob} computes a maximum $\textsc{d}$-\TOB  $\mc B$ with root $r$ for a given temporal graph when $\textsc{d}\in \{\LD, \ST\}$, and it is more involved with respect to Algorithm~\ref{alg:MTST}.
\begin{algorithm}[t]
\small
\caption{Computing a maximum $\textsc{d}$-\TOB, with $\textsc{d} \in \{\LD,\ST\}$.}\label{alg:LDST-tob}
\KwIn{A temporal graph $ \mathcal{G}=(V,A,\tau)$, a vertex $r\in V$, $\textsc{d}\in \{\LD,\ST\}$.}
\KwOut{A maximum ${\textsc{d}}$-\TOB $\mc B = (V\sub{\mc B},A\sub{\mc B},\tau\sub{\mc B})$ of $\mc G$ with root $r$.}
$\mc {EAD}(r)\gets 0; \forall v\in V\setminus \{r\}, \mc {EAD}(v)\gets +\infty $\;
$d(r)\gets 0$; $\forall v\in V\setminus \{r\},d(v)\gets \textsc{d}\sub{\mc G}(r,v)$\;
$\langle d_1,\dots, d_h\rangle \gets$ ordered list of finite $d$ values with no repetitions\;
$V\sub{\mc B}\gets \{r\}$; $A\sub{\mc B}\gets \emptyset$; $\tau\sub{\mc B}\gets 0$, $D_{0}\gets \{r\}$\;
\For{$i=1,\dots, h$}
{
$D_i\gets\{v\in V\setminus\{r\}: d(v)=d_i  \}  $\;
\lIf{$\textsc{d}=\LD$}{enqueue all $(r,v,s,t)\in A$ such that $ s=d_i$ in a min priority queue $Q$ with weight $t$}
\lIf{$\textsc{d}=\ST$}{enqueue all $(u,v,s,t)\in A$ such that $ u\in D_0\cup\ldots\cup D_{i-1}$ and $ v\in D_i$ in a min priority queue $Q$ with weight $t$}
    \While{$Q\neq \emptyset$}
    {
    dequeue $a\gets (u,v,s,t)$ from $Q$\;
    \While{$s< \mc{EAD}(u)$ or $ t\geq \mc{EAD}(v) $ or $(\textsc{d}=\ST\ and\ t-s\neq d_i-d(u))$}
        {
        \lIf{$Q=\emptyset$}{go to Line 5 with next value of $i$}
        dequeue $a\gets (u,v,s,t)$ from $Q$\;
        }
    \tcc{$a=(u,v,s,t)$ is s.t. $a\in \argmin_{(u',v',s',t')\in Q} t'$,\ $s\geq \mc{EAD}(u)$, $ t< \mc{EAD}(v)=+\infty $, and if $\textsc{d}=\ST$, $t-s= d_i-d(u)$.}
    $\mc {EAD}(v)\gets t$, $V\sub{\mc B}\gets V\sub{\mc B}\cup  \{v\}$, $A\sub{\mc B}\gets A\sub{\mc B}\cup  \{a\}$, $\tau\sub{\mc B}\gets \max\{\tau\sub{\mc B},t\}$\;
    enqueue all $(v,v',s',t')\in A$ such that $ v'\in D_i$ in $Q$ with weight $t'$\;
    }
}
\end{algorithm}
 The issue is that if 
$W\!=\!(r\!=\!v_0,a_1,\ldots,a_k,v_k\!=\!v)$ is a $\textsc{d}$-prefix-optimal walk, then it is possible to have $\textsc{d}(r,v_{i-1})\!=\!\textsc{d}(r,v_{i})$ for some $i\!\in\! [k]$. Indeed, if $\textsc{d}=\LD$, then all the vertices in the walk share the same latest departure time, i.e. $t_s(W)=\LD(r,v_i)$ for all $i\!\in \![k]$ and $\textsc{d}=\ST$ and $el(a_i)=0$ for some $i\!\in\! [k]$, then $\ST(r,v_{i-1})=\ST(r,v_i)$. However, in any case we have that $\textsc{d}(r,v_{i-1})\leq \textsc{d}(r,v_{i}) $ for all $i\!\in\! [k]$.
This implies that, by letting $D_i$ denote the set of vertices at distance $d_i$ from $r$ with the distances $d_i$ being in increasing order, to choose the parent of each vertex of $D_i$ in ${\mc B}$, we cannot look only at vertices in $D_0\cup\cdots \cup D_{i-1}$, but also at the ones in $D_i$ itself (in particular, only at the ones in $D_i$ when $\textsc{d}=\LD$). Note that this gives us an additional difficulty as we cannot simply choose an arbitrary vertex $v\in D_i$ to be the next one to be added to ${\mc B}$, as it might happen that the good parent of $v$ (i.e.\ the in-neighbor of $v$ within an $\EA{\textsc{d}}$-prefix-optimal $(r,v)$-walk) has not been added to ${\mc B}$ yet. To overcome this, we add vertices in $D_i$ to $\mc B$ in increasing order of the value of $\EA{\textsc{d}}(r,v)$. Observe however that $\EA{\textsc{d}}(r,v)$ is not known a priori, so to do that we use a queue that keeps the outgoing temporal arcs from vertices in ${\mc B}$ in increasing order of their arrival time. These ideas are formalized below.
At step $i$ of the \textbf{for} loop at lines 5-18,  Algorithm \ref{alg:LDST-tob} adds to ${\mc B}$ the vertices of $D_i$ that are reachable by a ${\textsc{d}}$-prefix-optimal walk. To this aim, it uses a min priority queue $Q$ for temporal arcs $a$ with head vertices in $D_i$ with weight $t_a(a)$. For $\textsc{d}=\LD$, $Q$ is initialized with all the outgoing temporal arcs from $r$ with starting time $d_i$, as they are the only arcs that can realize a latest departure time equal to $d_i$. For $\textsc{d}=\ST$, $Q$ is initialized with 
all the temporal arcs with tail in $D_0\cup\ldots\cup D_{i-1}$ and head in $D_i$. The vector $\mc{EAD}$  in the algorithm, initialized at $+\infty$ for all the vertices but the root, keeps track of the arrival time in the vertices every time they are added to the \TOB.
In the \textbf{while} loop at lines 9-17, we dequeue temporal arcs from $Q$ that cannot possibly be within an $\EA{\textsc{d}}$-prefix-optimal walk. Formally, if such loop is not broken in line 12, then at the end we are left with an arc $a=(u,v,s,t)\in \argmin\limits_{(u',v',s',t')\in Q} t'$, i.e. an arc that minimizes the arrival time in the queue, satisfying:
\begin{itemize}
\item $s\geq \mc{EAD}(u)$, so that $a$ is temporally compatible with the $(r,u)$-walk $W_u$ that is already present in the \TOB, i.e.\ $W_u+(u,a,v)$ is a temporal walk;
\item $ t< \mc{EAD}(v) $, which ensures that we add to the \TOB a new vertex each time; 
\item $t-s= d_i-d(u)$ if $\textsc{d}=\ST$, ensuring that $W_u+(u,a,v)$ realizes $\ST(r,v)$.
\end{itemize}
We then add $v$ and the temporal arc $a$ to the \TOB and we update the arrival time in $v$ to $\mc {EAD}(v)=t$, which is equal to $\EA\textsc{d}(r,v) $ and it will be no longer updated until the end of the algorithm. Finally, we add to $Q$ all the outgoing arcs from $v$ with head vertices in $D_i$.
When at distance $d_i$ there are no arcs satisfying these constraints, i.e.\ the queue $Q$ at line 12 is empty, we go to the next distance $d_{i+1}$, as it means that we have already spanned all the possible vertices in $D_i$. 
The initial computation of all $\textsc{d}(r,v)$ requires $O(m\log m)$ by Table \ref{table:shortest_path}. Concerning the remaining part of the algorithm, the $i$-th iteration of the \textbf{for} loop considers only arcs whose head is in $D_i$, hence each arc is considered only in one of the iterations of the \textbf{for} loop. Moreover, each arc is dequeued from $Q$ at most once. As the dequeue from $Q$ costs $O(\log m)$ we obtain a running time of $O(m\log m)$.

\begin{theorem}\label{thm:algLDST}
For any $\textsc{d}\!\in\! \{\LD,\ST \}$, Algorithm $\ref{alg:LDST-tob}$ returns a maximum $\textsc{d}$-\TOB of a temporal graph, for a chosen root, in $O(m\log m)$ time. Additionally, the output is an $\EA\textsc{d}$-\TOB.
\end{theorem}

\begin{proof}
Let $\mc G=(V,A,\tau)$ be the temporal graph input of the algorithm, $r\in V$, $d(v)=\textsc{d}\sub{\mc G}(r,v)$ for all $v\!\in\! V$, $h=| \{d(v):v\in V \}|$ and $   \{d_0,d_1,\dots ,d_h\}=\{d(v):v\in V, d(v)\!<\! +\infty \} $, with $d_0<d_1<\dots <d_h$. Let $V'=\{v\in V: v \text{ is temporally reachable from $r$} \}$ and for all $i\in [h]_0$, let $D_i=\{v\in V: d(v)=d_i  \}$. Note that $\{D_i:i\in  [h]_0\}$ is a partition of $V'$ with $D_0=\{r\}$. Since no confusion is possible, from now on we will avoid writing the subscripts $\mc G$.
Note that if $\textsc{d}=\LD$, then each iteration of the \textbf{for} loop in lines 5-18 is completely independent on the others, as it deals only with vertices in $D_i$ and temporal arcs with both tail and head in $D_i$. 
We now proceed by proving the following loop invariant:

\begin{claim}\label{claim:LDSTalg2}
Given $\textsc{d}\in \{\LD,\ST\}$, at the end of the $i$-th iteration of the \textbf{for} loop in lines 5-18, we have that $\mc{EAD}(v)=\EA\textsc{d}\sub{\mc G}(r,v)  $ for all $v\in V\sub{\mc B}$ and
$\mc B=(V\sub{\mc B},A\sub{\mc B},\tau\sub{\mc B})$ is an $\EA\textsc{d}$-\TOB with root $r$ of $\mc G$ with
\begin{equation}\label{eq:claimLDST}
    V\sub{\mc B}=\{v\in V': \exists \text{ a $\textsc{d}$-prefix-optimal temp. } (r,v)\text{-walk in } \mc G, d(v)\leq d_{i}  \}.
\end{equation}
\end{claim}
\noindent The claim above implies that the final output $\mc B$ of the algorithm is an $\EA\textsc{d}$-\TOB with root $r$ of $\mc G$, which is in particular a $\textsc{d}$-\TOB. Moreover, $V\sub {\mc B}$ consists of all the vertices in $\mc G$ for which there exists a $\textsc{d}$-prefix-optimal temporal walk from the root, so $\mc B$ is a maximum $\textsc{d}$-\TOB by Corollary \ref{cor:maxDtob}. We are left to prove the claim. 
$\mc B$ is initialized as the temporal graph made of the sole vertex $r$, so the loop invariant is trivially true. 
Suppose now that the loop invariant is true up to a certain $i$-th iteration. We now prove that it holds for the $(i+1)$-th iteration.
We start by proving that if $v\in D_{i+1}$ and there exists a $\textsc{d}$-prefix-optimal temporal $(r,v)$-walk in $\mc G$, say $W=(r=x_0,a_1,x_1,\dots ,a_{m},x_{m}=v)$, then $v\in V\sub{\mc B}$ at the end of the $(i+1)$-th \textbf{for} loop iteration. 
Let $x_j$ be the last vertex of $W$ starting from $r$ that is in $V\sub{\mc B}$ before the beginning of the $(i+1)$-th iteration ($x_j$ possibly equal to $r$). By inductive hypothesis, this implies that $d(x_j)<d_{i+1} $ and that $d(x_l)=d_{i+1}$ for all $l>j$. Then the arc $a_{j+1}$ is added to $Q$ in lines 7-8 at the beginning of the $(i+1)$-th iteration. Since 
$W$ is $\textsc{d}$-prefix-optimal,  $a_{j+1}$ does not fulfil the condition in line 11, unless $x_{j+1}$ has been already added to $V\sub{\mc B}$. This implies that at one point of the \textbf{while} loop $x_{j+1}$ is added to $V\sub{\mc B} $, which implies that $a_{j+2}$ is put in queue $Q$ by line 16. This iteratively proves that for every $l>j $, $x_l$  is added to $V\sub{\mc B}$ at one point of the \textbf{while} loop, including $x_{m}=v$. This proves equation (\ref{eq:claimLDST}).
To prove the rest of the claim, we will prove the following fact:
\begin{claim}\label{claim:LDSTalg}
Suppose to be in the $(i+1)$-th iteration of the \textbf{for} loop of lines 5-18. Then, at the end of each iteration of the \textbf{while} loop of lines 9–17, we have that $\mc B=(V\sub{\mc B},A\sub{\mc B},\tau\sub{\mc B})$ is an $\EA\textsc{d}$-\TOB with root $r$ and for all $ v\in V\sub{\mc B}, \,\mc{EAD}(v)=\EA\textsc{d}\sub{\mc G}(r,v)  $.
\end{claim}
\noindent At the beginning of the $(i+1)$-th \textbf{for} loop iteration, the inductive hypothesis holds, so the invariant property is true. 
By contradiction, consider the first iteration of the \textbf{while} loop such that the addition of the vertex $v$ to $V\sub{\mc B}$ and of the arc $a=(u,v,s,t)$ to $A\sub{\mc B}$ makes the claim fail. Since we are at the $(i+1)$-th \textbf{for} loop iteration, it holds that $d(v)=d_{i+1}$.
Due to line 11, it must hold that $s\geq \mc {EAD} (u)$ and, if $\textsc{d}=\ST$, then $t-s=d_{i+1}-d(u) $. 
This implies that $\mc {EAD} (u)<+\infty$, and since the only way for this to hold is to have $u=r$, or to have $\mc {EAD}(u)$ updated to a natural number (in which case $u$ is added to $V\sub {\mc B}$ in line 15),  we get that $u\in V\sub {\mc B}$. 
Also, $u$ must have entered $V\sub {\mc B}$ before $v$, so by hypothesis there exists an $\EA\textsc{d}$-prefix-optimal temporal $(r,u)$-walk $W_u$ in $\mc B$; in particular $t_a(W_u)=\EA\textsc{d}(r,u)= \mc{EAD}(u) $. Since $s\geq \mc {EAD} (u)$, then $W_v=W_u+(u,a,v)$ is a temporal $(r,v)$-walk in $\mc B$. Moreover, if $\textsc{d}=\LD$, then $u\in D_{i+1}$, so $t_s(W_v)=t_s(W_u)=d(u)=d_{i+1}=d(v)$. If $\textsc{d}=\ST$, then $ t-s=d_{i+1}-d(u)$, so $\weight(W_v)=\weight(W_u)+(t-s)=d(u)+(t-s)=d_{i+1}$. Hence in both cases $W_v$ is $\textsc{d}$-prefix-optimal. This also implies that $t=t_a(W_v)\geq \EA\textsc{d}(r,v)$. 
It remains to show that $v$ has indegree $1$ in $\mc B$ and that $t=\mc{EAD}(v)=\EA\textsc{d}(r,v)  $ to derive the contradiction.
Suppose first that $v$ has indegree $\neq 1$. Then it must have indegree greater than $1,$ because $a\in A\sub {\mc B}$ is an incoming temporal arc of $v$. Then, at a previous step of the \textbf{while} loop, an arc $a'=(u',v,s',t')$ was added to $A\sub {\mc B}$, which means that $v$ was also added to $V\sub{\mc B}$ and $\mc {EAD}(v)$ was set equal to $t'$.  
When $a'$ was added, $u'$ must have already been in $V\sub {\mc B}$. By hypothesis, there exists an $\EA\textsc{d}$-prefix-optimal temporal $(r,u')$-walk $W_{u'}$ in $\mc B$, and $W'=W_{u'}+(u',a',v)$ is such that $\EA\textsc{d}(r,v)=t_a(W')=t'=\mc{EAD}(v)$ by hypothesis. Since $t\geq \EA\textsc{d}(r,v)=\mc{EAD}(v)$, the arc $a$ could have never been chosen later, as it is fulfilling the condition $t\geq \mc{EAD}(v) $ in line 11. So $v$ has indegree $ 1$ in $\mc B$. 
Suppose now that $\mc{EAD}(v)\neq \EA\textsc{d}(r,v)  $, i.e.\ that $t>\EA\textsc{d}(r,v)$.
We know that $v$ has a $\textsc{d}$-prefix-optimal $(r,v)$-walk in $\mc G$; let $W$ the one that arrives the earliest in $v$, i.e. $t_a(W)=\EA\textsc{d}(r,v)$. 
Let $y\in V(W)$ be the first vertex along $W$ such that, when $v$ is added to $V\sub{\mc B}$, $y\notin V\sub {\mc B}$, and let $x\in V\sub {\mc B}$ be $y$'s predecessor along $W$ and $a_{xy}=(x,y,s_{xy},t_{xy})$ be the temporal arc connecting them in $W$ ($x$ may coincide with $r$). By inductive hypothesis, we have that $y\in D_{i+1} $.
Notice that $t_{xy}\leq t_a(W)=\EA\textsc{d}(r,v)$.
Since $x\in V\sub {\mc B}$ and we chose $v$ as the first vertex for which $\mc{EAD}(v)\neq \EA\textsc{d}(r,v)$, we have that $\mc{EAD}(x)=\EA\textsc{d}(r,x) $ when $x$ was added to $V\sub {\mc B}$. 
This implies that the arc $a_{xy}$ is enqueued in $Q$ when $x$ is added to $V\sub{\mc B}$.  
Indeed if $\textsc{d}=\LD$, then $d(x)=d(y)=d(v)=d_{i+1}$ and so $a_{xy}\in \{(x,v',s',t')\in A: v'\in D_{i+1} \}$ in line 16 and if $\textsc{d}=\ST$, let $i'$ such that $d_{i'} = d(x)\leq d_{i+1} $. If $i'=i+1$ we conclude as above. If $i'<i+1$, since $y\in D_{i+1}$, then $a_{xy}$ is enqueued in $Q$ at the beginning of the $(i+1)$-th \textbf{for} loop iteration (line 8). 
We claim that when $v$ was added to $V\sub {\mc B}$, $a_{xy}$ was still in $Q$. Indeed, $a_{xy}$ could have not been dequeued from $Q$ and added to $A\sub{\mc B}$ since otherwise $y\in  V\sub {\mc B}$ before $v$ was added to $V\sub {\mc B}$, which contradicts the hypothesis. If $a_{xy}$ was dequeued from $Q$ without being added to $A\sub{\mc B}$, since $W$ is $\textsc{d}$-prefix-optimal (and so $t_{xy}-s_{xy}=d(y)-d(x)=d_{i+1}-d(x) $ if $\textsc{d}=\ST$) and $s_{xy}\geq \EA\textsc{d}(r,x)=\mc{EAD}(x)$, then it must have hold that $\mc {EAD}(y)<+\infty$. This implies that $y$ was already in $V\sub {\mc B}$ before $v$ was added to $V\sub {\mc B}$, which again contradicts the hypothesis.
Therefore, when $v$ was added to $V\sub {\mc B}$,
it must hold that $t\leq t_{xy}$. But $t_{xy}\leq t_a(W)=\EA\textsc{d}(r,v)\leq t $ and so $ t=\EA\textsc{d}(r,v) $. This concludes the proof. 

 Regarding the computational complexity, the initial computation of all $\textsc{d}(r,v)$, $v\!\in\! V$, requires $O(m\log m)$ by Table \ref{table:shortest_path}. Concerning the remaining part of the algorithm, notice that the $i$-th iteration of the \textbf{for} loop considers only arcs whose head is in $D_i$. This means that each arc is considered only in one of the iterations of the \textbf{for} loop. Moreover, each arc is dequeued from $Q$ at most once. As the dequeue from $Q$ costs $O(\log m)$ we obtain a total running time of $O(m\log m)$.
\end{proof}

\section{Finding a maximum {\sc{d}}-temporal out-branching is \NP-hard for fastest time and minimum waiting time distances}\label{sec:FT-tobs}

As previously observed, Theorem \ref{thm:tob} does not hold for $\textsc{d}\in \{\FT, \MW\}$. Indeed in these cases the problem becomes \NP-complete even in the following very constrained situations:
when $el(a)=0$ for all $a\in A$, also called \emph{nonstrict} temporal graphs, and when $el(a)=1$ 
for all $a\in A$, also called \emph{strict} temporal graphs (see e.g.\ \cite{casteigts2018finding}). 
The nonstrict model is used when the time-scale of the measured phenomenon is relatively big: this is the case in a disease-spreading scenario~\cite{ZSCHOCHE2020,DELIGKAS2022} where the spreading speed might be unclear, or in time-varying graphs \cite{nicosia2012}, where a single snapshot corresponds e.g.\ to all the streets available within a day.\footnote{The literature often focused on nonstrict/strict variations to provide stronger negative results. In this paper, we have used the more general model to provide stronger positive results, while using the nonstrict/strict when providing negative ones.}
Notice that in the following theorem we consider we consider a particular case of the decision problems, namely, the existence of a spanning $\textsc{d}$-\TOB.

\begin{theorem}\label{thm:FT_tob_hard}
Let $\mc G=(V,A,\tau)$ be a temporal graph, $r\!\in\! V$ and $\textsc{d}\in \{\FT, \MW\}$. Deciding whether $\mc G$ has a spanning \textsc{d}-\TOB with root $r$ is \emph{\NP}-complete, even if $\tau=2$ and $el(a)\!=\!0$ for every $a\in A$, or if $\tau=4$ and $el(a)\!=\!1$ for every $a\!\in\! A$.
\end{theorem}

\begin{proof}
The problem is in \NP, since computing $\textsc{d}\sub{\mc G}(r,v)$ for every vertex $v$ can be done in polynomial time (Table \ref{table:shortest_path}), as well as testing whether a given temporal subgraph $\mc B$ is a $\textsc{d}$-\TOB (see Lemma \ref{lem:tob}). 
To prove hardness, we make a reduction from 3-SAT, largely known to be \NP-complete~\cite{cook1971complexity,levin1973universal}. For this, consider a formula $\phi$ in CNF form on variables $X = \{x_1,\ldots,x_n\}$ and on clauses $C = \{c_1,\ldots,c_m\}$. We first construct $\mc G = (V,A,\tau)$ for the case where every arc has elapsed time~0 (observe Figure~(\ref{fig:reduction0}) to follow the construction). Let $V = X \cup C\cup \{r\}$. For each variable $x_i$, add to $A$ the temporal arcs $(r,x_i,1,1)$ and $(r,x_i,2,2)$. Then, for each clause $c_j$ and each variable $x_i$ appearing in $c_j$, add temporal arc $(x_i,c_j,1,1)$ if $x_i$ appears in $c_j$ positively, while add the temporal arc $(x_i,c_j,2,2)$ if $x_i$ appears in $c_j$ negatively. We now prove that $\phi$ is satisfiable if and only if there exists a spanning \textsc{d}-\TOB rooted in $r$. 
Suppose first that $\phi$ has a satisfying assignment; we show how to construct a spanning \textsc{d}-\TOB ${\mc B} = (V,A\sub{\mc B},\tau\sub{\mc B})$ rooted in $r$. For each variable $x_i$, add to $A\sub{\mc B}$ the temporal arc $(r,x_i,1,1)$ if $x_i$ is true, while add to $A\sub{\mc B}$ the temporal arc $(r,x_i,2,2)$ if $x_i$ is false. Now consider a clause $c_j$ and choose one of the variables that validates $c_j$, say $x_{i_j}$. Add to $A\sub{\mc B}$ the unique temporal arc with head $c_j$ and tail $x_{i_j}$.  Now observe that the vertices in $X$ are connected to $r$ in ${\mc B}$ through direct arcs; hence we get that $\textsc{d}\sub{\mc B}(r,x_i) = 0$ for every $x_i\in X$. For a clause $c_j$, if $x_{i_j}$ appears positively in $c_j$, then $x_{i_j}$ is true, and $(r,x_{i_j},1,1)$ and $(x_{i_j},c_j,1,1)$ are in $A\sub{\mc B}$; therefore $\textsc{d}\sub{\mc B}(r,c_j) = 0$. If $x_{i_j}$ appears negatively in $c_j$, then $x_{i_j}$ is false, so $(r,x_{i_j},2,2)$ and $(x_{i_j},c_j,2,2)$ are in $A\sub{\mc B}$; therefore $\textsc{d}\sub{\mc B}(r,c_j) = 0$. 
Finally, observe that each vertex different from the root has indegree~1. By Lemma~\ref{lem:tob}, we get that ${\mc B}$ is a spanning \TOB, and since $\textsc{d}\sub{\mc B}(r,v) = 0$ for every $v\in V$, it follows that ${\mc B}$ is a spanning \textsc{d}-\TOB.

Suppose now that ${\mc B} = (V,A\sub{\mc B},\tau\sub{\mc B})$ is a spanning \textsc{d}-\TOB rooted in $r$. Since the only possible $(r,x_i)$-walk is through an arc, we get that either $(r,x_i,1,1)\in A\sub{\mc B}$ or $(r,x_i,2,2)\in A\sub{\mc B}$. If the former occurs, then set $x_i$ to true, while if the latter occurs, then set $x_i$ to false. We now argue that this must be a satisfying assignment to $\phi$. For this, consider a clause $c_j$. By Lemma~\ref{lem:tob}, we know that $d^-\sub{\mc B}(c_j)=1$; so let $a = (x_{i_j},c_j,t,t)$ be the temporal arc incident to $c_j$ in ${\mc B}$. If $x_{i_j}$ appears positively in $c_j$, then we know that $a=(x_{i_j},c_j,1,1)$ by construction. And since the temporal $(r,c_j)$-walk must pass by $x_{i_j}$, we get that $(r,x_{i_j},1,1)\in A_{\tau}$, in which case $x_{i_j}$ is set to true and hence satisfies $c_j$. If $x_{i_j}$ appears in $c_j$ negatively, then $a=(x_{i_j},c_j,2,2)$. Notice that $\textsc{d}\sub{\mc G}(r,c_j) = 0$; since ${\mc B}$ is a \textsc{d}-\TOB, we must also have $\textsc{d}\sub{\mc B}(r,c_j) = 0$. This implies that $(r,x_{i_j},2,2)\in A_{\mc B}$ and hence $x_{i_j}$ is set to false, satisfying $c_j$.

In the case where $el(a)=1$ for every arc $a$, the reduction is similar to the previous one. Specifically, for each $x_i\in X$, we add arcs $(r,x_i,1,2)$ and $(r,x_i,2,3)$. For each clause $c_j$, if $x_i$ appears positively in $c_j$ we add the temporal arc $(x_i,c_j,2,3)$, while if $x_i$ appears negatively in $c_j$ we add the temporal arc $(x_i,c_j,3,4)$. Analogous arguments to the previous ones apply.
\end{proof}
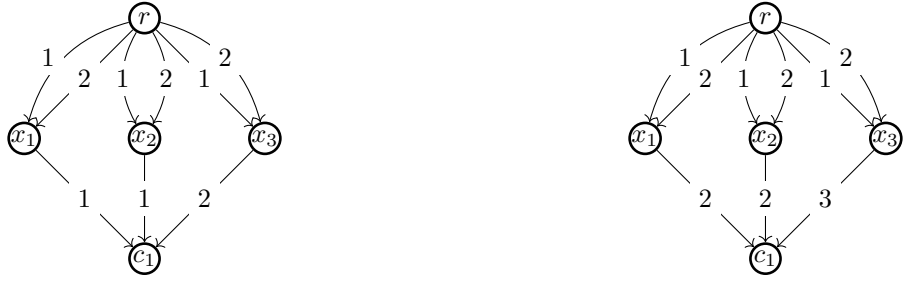
\begin{figure*}[t!]
    \centering
    \begin{subfigure}[t]{0.5\textwidth}
        \centering
        \begin{tikzpicture}[roundnode/.style={circle,fill=white, inner sep=0pt, minimum size=4mm,draw=black},node distance=0.5cm,line width=0.2mm, scale=0.8]
  \pgfsetlinewidth{1pt}
  \pgfdeclarelayer{bg}    
   \pgfsetlayers{bg,main}  


   			\node[roundnode]    (r) at (0,0) {$r$};
 			\node[roundnode]    (x1) at (-2,-2) {$x_1$};
 			\node[roundnode]    (x2) at (0,-2) {$x_2$};
 			\node[roundnode]    (x3) at (2,-2) {$x_3$};
 			\node[roundnode]    (c1) at (0,-4) {$c_1$};

    \begin{pgfonlayer}{bg}    
        \path[->] (r) edge[bend right,left]	node[fill=white]  {$1$} (x1)
 				  (r) edge 		node[fill=white]  {$2$} (x1)
 				  (r) edge[bend right]	node[fill=white]  {$1$} (x2)
 				  (r) edge[bend left] 		node[fill=white]  {$2$} (x2)
 				  (r) edge	node[fill=white]  {$1$} (x3)
 				  (r) edge[bend left] 		node[fill=white]  {$2$} (x3) 
 				  (x1) edge	node[fill=white]  {$1$} (c1)
 				  (x2) edge	node[fill=white]  {$1$} (c1)
 				  (x3) edge	node[fill=white]  {$2$} (c1);
    \end{pgfonlayer}
    
  \end{tikzpicture}
        \caption{All temporal arcs have elapsed time~0.}\label{fig:reduction0}
    \end{subfigure}%
    \begin{subfigure}[t]{0.5\textwidth}
        \centering
        \begin{tikzpicture}[roundnode/.style={circle,fill=white, inner sep=0pt, minimum size=4mm,draw=black},node distance=0.7cm,line width=0.2mm,scale=0.8]
  \pgfsetlinewidth{1pt}
  \pgfdeclarelayer{bg}    
   \pgfsetlayers{bg,main}  

   			\node[roundnode]    (r) at (0,0) {$r$};
 			\node[roundnode]    (x1) at (-2,-2) {$x_1$};
 			\node[roundnode]    (x2) at (0,-2) {$x_2$};
 			\node[roundnode]    (x3) at (2,-2) {$x_3$};
 			\node[roundnode]    (c1) at (0,-4) {$c_1$};

    \begin{pgfonlayer}{bg}    
        \path[->] (r) edge[bend right]	node[fill=white]  {$1$} (x1)
 				  (r) edge 		node[fill=white]  {$2$} (x1)
 				  (r) edge[bend right]	node[fill=white]  {$1$} (x2)
 				  (r) edge[bend left] 		node[fill=white]  {$2$} (x2)
 				  (r) edge	node[fill=white]  {$1$} (x3)
 				  (r) edge[bend left] 		node[fill=white]  {$2$} (x3)
 				  (x1) edge	node[fill=white]  {$2$} (c1)
 				  (x2) edge	node[fill=white]  {$2$} (c1)
 				  (x3) edge	node[fill=white]  {$3$} (c1);
    \end{pgfonlayer}

  \end{tikzpicture}
        \caption{All temporal arcs have elapsed time~1.}\label{fig:reduction1}
    \end{subfigure}
    \caption{Example of the construction in the proof of Theorem~\ref{thm:FT_tob_hard}. Clause $c_1$ is equal to $(x_1\vee x_2\vee\neg x_3)$. The value on top of each arc represents the starting time.}\label{fig:reduction}
\end{figure*}

\noindent The gaps left by the above theorem are when ${\mc G}$ has lifetime~1 or when ${\mc G}$ has lifetime $\tau\in \{2,3 \}$ and all arcs have elapsed time at least~1. Those cases are investigated in the following proposition.

\begin{proposition}\label{prop:gaps_FTMW}
  Let $\mc G=(V,A,\tau)$ be a temporal graph, $r\!\in\! V$ and $\textsc{d}\in \{\FT, \MW\}$. If $\tau=1$ or, if $\tau\in \{2,3 \}$ and $el(a)\geq 1$ for all $a\in A$, then a maximum \textsc{d}-\TOB with root $r$ of $\mc G$ is computable in polynomial time.
\end{proposition}

\begin{proof}
If $\tau=1$, the temporal graph reduces to a static graph, so the problem is solvable in polynomial time by Dijkstra's algorithm. Suppose now that $\tau=2$ and $el(a)\geq 1$ for all $a\in A$. Then a maximum \textsc{d}-\TOB rooted in $r$ contains exactly $r$ and every $u\in V$ such that $(r,u,1,2)$ is an arc in ${\mc G}$. Finally, let $\tau=3$ and $el(a)\geq 1$ for all $a\in A$. Note that each arc has a temporal label belonging to the set $\{(1,2),(2,3),(1,3) \}$. This implies that a temporal walk has length of at most $2$ and that every temporal walk realizing $\textsc{d}$ is also $\textsc{d}$-prefix-optimal, because it is either made of just one arc, or it is made by two arcs consecutively labeled by $(1,2)$ and $(2,3)$. Intuitively, to build a maximum \textsc{d}-\TOB, we first add all the vertices reachable from the root by an arc labeled by $(1,2)$, and secondly we add all the vertices (different from the previous ones) that are reached from the root by an arc labeled by $(1,3)$. Finally, we add all the vertices (not yet added) that are reachable from $r$ by a temporal path of length $2$.
More formally, we set $A_1=\{(r,v,1,2)\in A: v\in V\}$,
$V_1=\{v\in V: (r,v,1,2)\in A_1 \}$, $A_2=\{(r,v,2,3)\in A: v\in V\setminus V_1\}$, $V_2=\{v\in V: (r,v,2,3)\in A_2 \}$, $A_3=\{(r,v,1,3)\in A : v\in V\setminus (V_1\cup V_2)\}$,
$V_3=\{v\in V: (r,v,1,3)\in A_3\}$, $A_4=\{(u,v,2,3)\in A: u\in V_1, v\in V\setminus (V_1\cup V_2\cup V_3)\}$, $V_4=\{v\in V: u\in V_1,\,(u,v,2,3)\in A_4, \}$, $V\sub{\mc B}=\{r\}\cup V_1\cup V_2\cup V_3\cup V_4$ and $A\sub{\mc B}=A_1\cup A_2\cup A_3\cup A_4$.
Then a maximum \textsc{d}-\TOB for $\mc G$ rooted in $r$ is the subgraph $\mc B=(V\sub{\mc B},A\sub{\mc B},3 )$.
\end{proof}

\section{Reaching vertices with no prefix-optimal paths: minimum temporal spanning subgraphs}\label{sec:TSS}
We have seen that, for $\textsc{d}\in \{\LD,\MT,\ST \}$, when a vertex $v$ in a temporal graph does not have a $\textsc{d}$-prefix-optimal path from the root, then no $\textsc{d}$-\TOB reaching $v$ is possible. We can tackle the problem from another point of view, where we want to reach anyway all the vertices of the graph from the root optimizing some distance, while still using the least amount of connections possible. 

\begin{definition}\label{def:dsub}
Let $\mc G=(V,A,\tau)$ a temporal graph, $r\in V$ and $\textsc{d}$ a distance. Then $\mc G'=(V,A',\tau')$ is a \emph{$\textsc{d}$-Temporal Out-Spanning Subgraph {\rm(}$\textsc{d}$-\TSS{\rm)} with root $r$} of $\mc G$ if for all $ v\in V$, there exists a temporal $(r,v)$-walk in $\mc  G'$ realizing $ \textsc{d}\sub {\mc G}(r,v)$; if in addition $|A'|$ is the smallest possible, then $\mc G'$ is a \emph{minimum} $\textsc{d}$-\TSS.
\end{definition}
\noindent A spanning $\textsc{d}$-\TOB with root $r$ is always a minimum $\textsc{d}$-\TSS with root $r$. On the other hand, notice that the only $\textsc{d}$-\TSS with root $r$ of the temporal graphs in Figure \ref{fig:tob_notP}, for the corresponding distances, are the graph themselves. 
Clearly, for $\textsc{d}\in \{\LD,\MT,\ST \}$, if every vertex of the temporal graph has a $\textsc{d}$-prefix-optimal path from the root, then a minimum $\textsc{d}$-\TSS with root $r$ is a spanning $\textsc{d}$-\TOB with root $r$. 
We now focus on finding a minimum $\textsc{d}$-\TSS of a given temporal graph, for a chosen root and distance $\textsc{d}$. In particular, we study the computational complexity of the following problem:
\begin{problem}\label{pb:TOSS}
Given $\mc G=(V,A,\tau)$ a temporal graph, $r\in V$, $k\in \mathbb N$ and $\textsc{d}\in \{\EA,\FT, \LD,\MT,\MW,\ST \}$, decide whether $\mc G$ has a $\textsc{d}$-\TSS with root $r$ and with at most $k$ temporal arcs.  
\end{problem}
\noindent Notice that for $\textsc{d}=\EA$, the concepts of spanning $\EA$-\TOB and minimum $\EA$-\TSS coincide\footnote{Since a temporal graph $\mc G$ has a spanning $\EA$-\TOB with a given root as subgraph if and only if each vertex is reachable from the root in $\mc G$ \cite{huang2015}.}. This implies that Problem \ref{pb:TOSS} for $\textsc{d}=\EA$ is solvable in polynomial time: if some vertex is not temporally reachable from the root the answer is simply NO; if every vertex is temporally reachable from the root, then the answer is YES if and only if $k\geq |V|-1$ \cite{huang2015}. On the other hand, for $\textsc{d}\in \{\FT,\MW \}$, Problem \ref{pb:TOSS} becomes \NP-hard as it suffices to set $k=|V|-1$ and apply Theorem \ref{thm:FT_tob_hard}. 
For all the other distances, the problem turns out to be a difficult task also in very constrained situations.

\begin{theorem}\label{thm:TSS_hard}
Let $\mc G=(V,A,\tau)$ be a temporal graph, $r\in V$, $k\in \mathbb N$, and suppose that there exists $a\in A$ such that $ el(a)\neq 0$. Then deciding whether $\mc G$ has an $\ST$-\TSS with root $r$ and with at most $k$ temporal arcs
is \emph{\NP}-complete even when $\tau=3$.
\end{theorem}

\begin{proof}
Let $\mc G'=(V,A',\tau')$ be a temporal subgraph of $\mc G$. Computing $\ST\sub{\mc G}(r,v)$ and $\ST\sub{{\mc G'}}(r,v)$ for every vertex $v$ can be done in polynomial time (Table \ref{table:shortest_path}), as well as checking if $|A'|\leq k$, so the problem is in \NP.

\noindent To prove hardness we make a reduction from 3-SAT. 
Consider a formula $\phi$ in CNF form on variables $X = \{x_1,\ldots,x_l\}$ and clauses $C = \{c_1,\ldots,c_m\}$. We construct a temporal graph $\mc G = (V,A,3)$ in the following way. Let $V = \{x^p_1,\ldots,x^p_l\} \cup \{x^n_1,\ldots,x^n_l\} \cup \{y_1,\ldots,y_l\} \cup C\cup \{r\}$. Notice that $|V|=3l+m+1$.
For each $\alpha\in \{p,n\}$ and $i\in [l]$, add to $A$ the temporal arcs $(r,x^{\alpha}_i,1,2)$, $(r,x^{\alpha}_i,3,3)$ and $(x^{\alpha}_i,y_i,2,2)$. 
Then, for each clause $c_j$ and each variable $x_i$ appearing in $c_j$, add the temporal arc $(x^p_i,c_j,2,2)$ if $x_i$ appears in $c_j$ positively, while add the temporal arc $(x^n_i,c_j,2,2)$ if $x_i$ appears in $c_j$ negatively (see Figure \ref{fig:ST_dsubgraph}). Observe that $\ST\sub{ \mc G}(r,x^{\alpha}_i)=0$ for all $\alpha\in \{p,n\}$, $i\in [l]$, $\ST\sub{ \mc G}(r,y_i)=1$ for all $i\in [l]$ and $\ST\sub{ \mc G}(r,c_j)=1$ for all $j\in [m]$.
We now prove that $\phi$ is satisfiable if and only if there exists an $\ST$-\TSS $\mc G'=(V,A',3)$ of $\mc G$ with root $r$ such that $|A'|\leq 4l+m $.

\noindent Suppose first that $\phi$ has a satisfying assignment; we show how to construct $\mc G'$. 
For each variable $x_i$, if $x_i$ is true then add to $A'$ the temporal arcs $(r,x^p_i,1,2)$,$(r,x^p_i,3,3)$,$(r,x^n_i,3,3)$,$(x^p_i,y_i,2,2)$, while if $x_i$ is false then add to $A'$ the temporal arcs $(r,x^n_i,1,2)$,$(r,x^n_i,3,3)$,$(r,x^p_i,3,3)$,$(x^n_i,y_i,2,2)$. 
Now consider a clause $c_j$ and choose one of the variables that validates $c_j$, say $x_{i_j}$. Add to $A'$ the unique temporal arc with head $c_j$ and tail $x^{\alpha_{i_j}}_{i_j}$, $\alpha_{i_j}\in \{p,n \}$. It holds that $|A'|= 4l+m $.
Now observe that for all $v\in V$ there exists a temporal $(r,v)$-path in $\mc G'$ realizing $\ST\sub{\mc G}(r,v)$: indeed, the vertices $x^{\alpha}_i$ are directly connected to $r$ by the arcs with time labels $(3,3)$ realizing their shortest times. Vertices $y_i$ are connected to $r$ by the path $(r,x^p_i,1,2)$, $(x^p_i,y_i,2,2)$ if $x_i$ is true, or by the path $(r,x^n_i,1,2)$, $(x^n_i,y_i,2,2)$ if $x_i$ is false, which realizes $\ST\sub{\mc G}(r,y_i)$. 
For a clause $c_j$, if $x_{i_j}$ appears positively in $c_j$, then $x_{i_j}$ is true, and so $(r,x^p_{i_j},1,2)$, $(x^p_{i_j},c_j,2,2)$ is a temporal path in $\mc G'$ reaching $c_j$ and realizing $\ST\sub{\mc G}(r,c_j)$; if $x_{i_j}$ appears negatively in $c_j$, then $x_{i_j}$ is false, and so $(r,x^n_{i_j},1,2)$, $(x^n_{i_j},c_j,2,2)$ is a temporal path in $\mc G'$ reaching $c_j$ and realizing $\ST\sub{\mc G}(r,c_j)$.

\noindent Suppose now that ${\mc G'} = (V,A',3)$ is an $\ST$-\TSS of $\mc G$ with $|A'|\leq 4l+m $. Since $\mc G'$ is spanning, each vertex different from the root must have at least one incoming arc, so $|A'|\geq 3l+m=|V|-1 $. In particular, for all $i\in [l]$, $\alpha\in \{p,n \}$, the arc $(r,x^\alpha_{i},3,3) $ must belong to $ A'$ as it is the only one realizing $\ST(r,x^\alpha_{i})$. Then notice that the only way a vertex $y_i$ can be temporally reachable from $r$ is that a least one of the arcs $(r,x^p_{i},1,2)$ and $(r,x^n_{i},1,2)$ belongs to $A'$.
Since $|A'|\leq 4l+m $, it must hold that for every $i\in [l]$, exactly one of the arcs between $(r,x^p_{i},1,2)$ and $(r,x^n_{i},1,2)$ belongs to $A'$.
If the former case occurs, then set $x_i$ to true, while if the latter case occurs, then set $x_i$ to false. We now argue that this must be a satisfying assignment to $\phi$. For this, consider a clause $c_j$. Each clause $c_j$ must be reached exactly by one arc in $\mc G'$, as otherwise there would not be enough arcs to connect every vertex; let $(x^{\alpha_{i_j}}_{i_j},c_j,2,2)$ be this arc. This also implies that $(r,x^{\alpha_{i_j}}_{i_j},1,2)\in A'$, as otherwise $c_j$ would not be temporally reachable from $r$. 
If $\alpha_{i_j}=p$, then $x_{i_j} $ is set to true and,
by construction, $x_{i_j}$ appears positively in $c_j$; so $x_{i_j}$ satisfies $c_j$.
If $\alpha_{i_j}=n$, then $x_{i_j} $ is set to false and,
by construction, $x_{i_j}$ appears negatively in $c_j$; so $x_{i_j}$ satisfies $c_j$.
\end{proof}
\noindent The gaps left by the above theorem are when ${\mc G}$ has lifetime $\tau\in \{1,2\}$ or when $el(a)=0$ for all arcs $a$ of $\mc G$. When $\tau=1$ the temporal graph reduces to a static graph and the problem reduces to finding an out-branching of $\mc G$ and when $\tau=2$, by Lemma \ref{claim:STtau} the problem reduces to find a spanning \ST-\TOB , which is solvable in polynomial time. Finally, if $el(a)=0$ for all $a\in A$, then $\ST(r,v)=0$ for all $v\in V$; hence the problem reduces to finding a \TOB, which is then solvable in polynomial time (Table \ref{tab:results}).

\begin{figure}
    \centering
  \begin{tikzpicture}[scale=0.3]
\SetVertexStyle[FillColor=white]
\SetEdgeStyle[Color=black,LineWidth=0.5pt]
  \Vertex[y=3,x=0,label=r,color=white!70!black]{1}
  \Vertex[y=-2,x=-13,label=$x_1^p$]{2}
  \Vertex[y=-2,x=-9,label=$x_1^n$]{3}
\Vertex[y=-4.5,x=-2,label=$x_2^p$]{5}
  \Vertex[y=-4.5,x=2,label=$x_2^n$]{6}
  \Vertex[y=-2,x=9,label=$x_3^p$]{8}
  \Vertex[y=-2,x=13,label=$x_3^n$]{9}
  
  \Vertex[y=-6,x=-11,label=$y_1$]{4}
\Vertex[y=-8.5,x=0,label=$y_2$]{7}
\Vertex[y=-6,x=11,label=$y_3$]{10}

  \Vertex[y=-12,x=0,label=$c$]{11}
  
  \Edge[Direct,label={(1,2)},bend=-5,fontsize=\tiny](1)(2)
\Edge[Direct,label={(1,2)},fontsize=\tiny](1)(3)
\Edge[Direct,label={(1,2)},bend=-25,fontsize=\tiny](1)(5)
\Edge[Direct,label={(1,2)},fontsize=\tiny](1)(6)
\Edge[Direct,label={(1,2)},bend=10,fontsize=\tiny](1)(8)
\Edge[Direct,label={(1,2)},bend=10,fontsize=\tiny](1)(9)

 \Edge[Direct,label={(3,3)},bend=-25,fontsize=\tiny,distance=0.8](1)(2)
 \Edge[Direct,label={(3,3)},bend=20,fontsize=\tiny,distance=0.8](1)(3)
 \Edge[Direct,label={(3,3)},bend=10,fontsize=\tiny,distance=0.7](1)(5)
  \Edge[Direct,label={(3,3)},bend=35,fontsize=\tiny,distance=0.7](1)(6)
   \Edge[Direct,label={(3,3)},bend=-10,fontsize=\tiny,distance=0.8](1)(8)
    \Edge[Direct,label={(3,3)},bend=30,fontsize=\tiny,distance=0.8](1)(9)

  \Edge[Direct,label={(2,2)},fontsize=\tiny](2)(4)
  \Edge[Direct,label={(2,2)},fontsize=\tiny](3)(4)
   \Edge[Direct,label={(2,2)},fontsize=\tiny](5)(7)
  \Edge[Direct,label={(2,2)},fontsize=\tiny](6)(7)
   \Edge[Direct,label={(2,2)},fontsize=\tiny](8)(10)
  \Edge[Direct,label={(2,2)},fontsize=\tiny](9)(10)

 \Edge[Direct,label={(2,2)},bend=-50,fontsize=\tiny](2)(11)
  \Edge[Direct,label={(2,2)},bend=30,fontsize=\tiny](6)(11)
   \Edge[Direct,label={(2,2)},bend=50,fontsize=\tiny](9)(11)
 
\end{tikzpicture}
    \caption{Example of the construction in the proof of Theorem \ref{thm:TSS_hard}. Clause $c$ is equal to $(x_1\vee \neg x_2\vee\neg x_3)$.}
    \label{fig:ST_dsubgraph}
\end{figure}
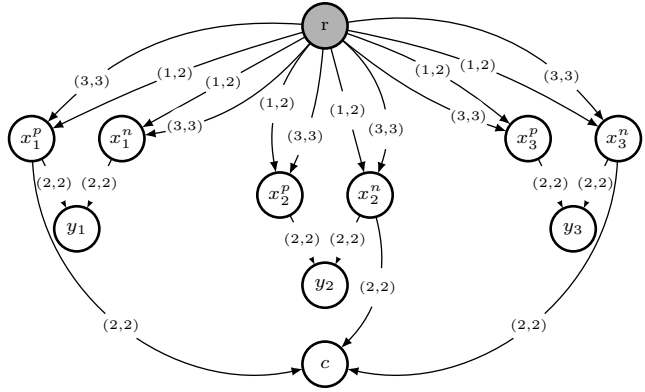

\begin{theorem}\label{thm:LD_min_hard}
Let $\mc G=(V,A,\tau)$ be a temporal graph, $r\in V$ and $k\in \mathbb N$. Then deciding whether $\mc G$ has an $\LD$-\TSS with root $r$ and with at most $k$ temporal arcs
is \emph{\NP}-complete even when $\tau=2$ and $el(a)= 0$ for every $a\in A$, or when $\tau=3$ and $el(a)= 1$ for every $a\in A$.
\end{theorem}

\begin{proof}
Let $\mc G'=(V,A',\tau')$ be a temporal subgraph of $\mc G$. Computing $\LD\sub{\mc G}(r,v)$ and $\LD\sub{{\mc G'}}(r,v)$ for every vertex $v$ can be done in polynomial time (Table \ref{table:shortest_path}), as well as checking if $|A'|\leq k$, so the problem is in \NP.

\noindent To prove hardness we make a reduction from 3-SAT, in a similar way to the proof of Theorem \ref{thm:TSS_hard}. 
Consider a formula $\phi$ in CNF form on variables $X = \{x_1,\ldots,x_l\}$ and on clauses $C = \{c_1,\ldots,c_m\}$. We construct a temporal graph $\mc G = (V,A,2)$ in the following way. Let $V = \{x^p_1,\ldots,x^p_l\} \cup \{x^n_1,\ldots,x^n_l\} \cup \{y_1,\ldots,y_l\} \cup C\cup \{r\}$. Notice that $|V|=3l+m+1$.
For each $\alpha\in \{p,n\}$ and $i\in [l]$, add to $A$ the temporal arcs $(r,x^{\alpha}_i,1,1)$, $(r,x^{\alpha}_i,2,2)$ and $(x^{\alpha}_i,y_i,1,1)$. 
Then, for each clause $c_j$ and each variable $x_i$ appearing in $c_j$, add the temporal arc $(x^p_i,c_j,1,1)$ if $x_i$ appears in $c_j$ positively, while add the temporal arc $(x^n_i,c_j,1,1)$ if $x_i$ appears in $c_j$ negatively (see Figure \ref{fig:reductionLD0}). Observe that $\LD\sub{ \mc G}(r,x^{\alpha}_i)=2$ for all $\alpha\in \{p,n\}$, $i\in [l]$, while  $\LD\sub{ \mc G}(r,v)=1$ for all the other vertices $v$. 
We now prove that $\phi$ is satisfiable if and only if there exists a $\LD$-\TSS of $\mc G$ with root $r$ with at most $4l+m $ arcs. Suppose that $\phi$ has a satisfying assignment; we show how to construct a $\LD$-\TSS  $\mc G'=(V,A',2)$ with root $r$ of $\mc G$ with $|A'|\leq 4l+m $.
For each variable $x_i$, if $x_i$ is true then add to $A'$ the temporal arcs $(r,x^p_i,1,1)$,$(r,x^p_i,2,2)$,$(r,x^n_i,2,2)$,$(x^p_i,y_i,1,1)$, while if $x_i$ is false then add to $A'$ the temporal arcs $(r,x^n_i,1,1)$,$(r,x^n_i,2,2)$,$(r,x^p_i,2,2)$,$(x^n_i,y_i,1,1)$. 
Now consider a clause $c_j$ and choose one of the variables that validates $c_j$, say $x_{i_j}$. Add to $A'$ the unique temporal arc with head $c_j$ and tail $x^{\alpha_{i_j}}_{i_j}$, $\alpha_{i_j}\in \{p,n \}$. It holds that $|A'|= 4l+m $.
Observe that for all $v\in V$ there exists a temporal $(r,v)$-path in $\mc G'$ realizing $\LD\sub{\mc G}(r,v)$. 

\noindent Suppose now that ${\mc G'} = (V,A',2)$ is an $\LD$-\TSS of $\mc G$ with $|A'|\leq 4l+m $. Since $\mc G'$ is spanning, each vertex different from the root must have at least one incoming arc, so $|A'|\geq 3l+m=|V|-1 $. In particular, for all $i\in [l]$, $\alpha\in \{p,n \}$, the arc $(r,x^\alpha_{i},2,2) $ must belong to $ A'$ since it is the only one realizing $\LD(r,x^\alpha_{i})$. Then notice that the only way a vertex $y_i$ can be temporally reachable from $r$ is that at least one of the arcs $(r,x^p_{i},1,1)$ and $(r,x^n_{i},1,1)$ belongs to $A'$.
Since $|A'|\leq 4l+m $, it must hold that for every $i\in [l]$, exactly one of the arcs between $(r,x^p_{i},1,1)$ and $(r,x^n_{i},1,1)$ belongs to $A'$.
If the former case occurs, then set $x_i$ to true, while if the latter case occurs, then set $x_i$ to false. We now argue that this must be a satisfying assignment to $\phi$. For this, consider a clause $c_j$. Each clause $c_j$ must be reached exactly by one arc in $\mc G'$, as otherwise there would not be enough arcs to connect every vertex; let $(x^{\alpha_{i_j}}_{i_j},c_j,1,1)$ be this arc. This also implies that $(r,x^{\alpha_{i_j}}_{i_j},1,1)\in A'$, as otherwise $c_j$ would not be temporally reachable from $r$. 
If $\alpha_{i_j}=p$, then $x_{i_j} $ is set to true and,
by construction, $x_{i_j}$ appears positively in $c_j$; so $x_{i_j}$ satisfies $c_j$.
If $\alpha_{i_j}=n$, then $x_{i_j} $ is set to false and,
by construction, $x_{i_j}$ appears negatively in $c_j$; so $x_{i_j}$ satisfies $c_j$.

In the case where $el(a)=1$ for every arc $a$, the reduction is similar to the previous one. Specifically, for each $i\in [l]$, $\alpha\in \{p,n \}$, we add arcs $(r,x^\alpha_i,1,2)$, $(r,x^\alpha_i,2,3)$, $(x^\alpha_i,y_i,2,3)$. For each clause $c_j$, if $x_i$ appears positively in $c_j$ we add the temporal arc $(x^p_i,c_j,2,3)$, while if $x_i$ appears negatively in $c_j$ we add the temporal arc $(x^n_i,c_j,2,3)$; see Figure \ref{fig:reductionLD1}. Analogous arguments to the previous ones apply.
\end{proof}
\noindent The gaps left by the above theorem are when ${\mc G}$ has lifetime~1 or when ${\mc G}$ has lifetime 2 and all arcs have elapsed time at least~1. In the first case, the temporal graph reduces to a static graph, so the problem is solvable in polynomial time by Dijkstra's algorithm. As for the second case, if all vertices of $\mc G$ are temporally reachable from the root $r$ one can see that a \LD-\TSS rooted in $r$ contains exactly all the arcs of type $(r,v,1,2)$, for every $v\in V$. 

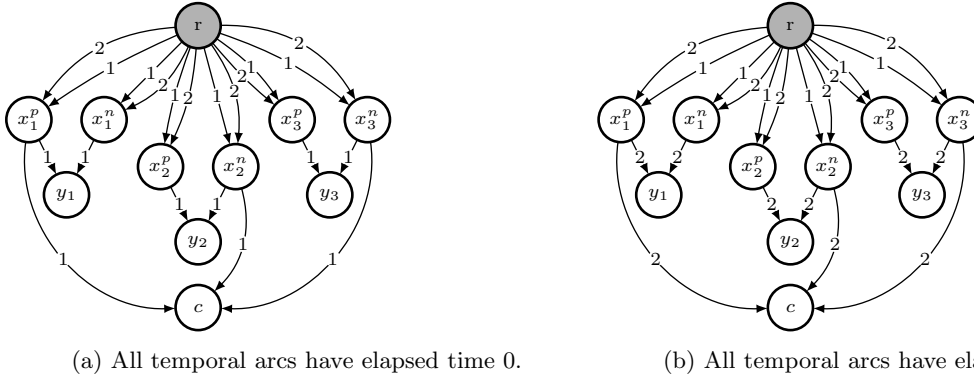
\begin{figure}
    \centering
     \begin{subfigure}[t]{0.47\textwidth}
    \begin{tikzpicture}[scale=0.25]
\SetVertexStyle[FillColor=white]
\SetEdgeStyle[Color=black,LineWidth=0.5pt]
  \Vertex[y=3,x=0,label=r,color=white!70!black]{1}
  \Vertex[y=-2,x=-9,label=$x_1^p$]{2}
  \Vertex[y=-2,x=-5,label=$x_1^n$]{3}
\Vertex[y=-4.5,x=-2,label=$x_2^p$]{5}
  \Vertex[y=-4.5,x=2,label=$x_2^n$]{6}
  \Vertex[y=-2,x=5,label=$x_3^p$]{8}
  \Vertex[y=-2,x=9,label=$x_3^n$]{9}
  
  \Vertex[y=-6,x=-7,label=$y_1$]{4}
\Vertex[y=-8.5,x=0,label=$y_2$]{7}
\Vertex[y=-6,x=7,label=$y_3$]{10}

  \Vertex[y=-12,x=0,label=$c$]{11}
  
  \Edge[Direct,label={1},bend=-5](1)(2)
\Edge[Direct,label={1}](1)(3)
\Edge[Direct,label={1},bend=-5](1)(5)
\Edge[Direct,label={1}](1)(6)
\Edge[Direct,label={1},bend=10](1)(8)
\Edge[Direct,label={1},bend=10](1)(9)

 \Edge[Direct,label={2},bend=-25](1)(2)
 \Edge[Direct,label={2},bend=20](1)(3)
 \Edge[Direct,label={2},bend=10](1)(5)
  \Edge[Direct,label={2},bend=20](1)(6)
   \Edge[Direct,label={2},bend=-5](1)(8)
    \Edge[Direct,label={2},bend=30](1)(9)

  \Edge[Direct,label={1}](2)(4)
  \Edge[Direct,label={1}](3)(4)
   \Edge[Direct,label={1}](5)(7)
  \Edge[Direct,label={1}](6)(7)
   \Edge[Direct,label={1}](8)(10)
  \Edge[Direct,label={1}](9)(10)

 \Edge[Direct,label={1},bend=-50](2)(11)
  \Edge[Direct,label={1},bend=30](6)(11)
   \Edge[Direct,label={1},bend=50](9)(11)
\end{tikzpicture}
\caption{All temporal arcs have elapsed time~0.}\label{fig:reductionLD0}
\end{subfigure}
 \begin{subfigure}[t]{0.47\textwidth}
   \begin{tikzpicture}[scale=0.25]
\SetVertexStyle[FillColor=white]
\SetEdgeStyle[Color=black,LineWidth=0.5pt]
  \Vertex[y=3,x=0,label=r,color=white!70!black]{1}
  \Vertex[y=-2,x=-9,label=$x_1^p$]{2}
  \Vertex[y=-2,x=-5,label=$x_1^n$]{3}
\Vertex[y=-4.5,x=-2,label=$x_2^p$]{5}
  \Vertex[y=-4.5,x=2,label=$x_2^n$]{6}
  \Vertex[y=-2,x=5,label=$x_3^p$]{8}
  \Vertex[y=-2,x=9,label=$x_3^n$]{9}
  
  \Vertex[y=-6,x=-7,label=$y_1$]{4}
\Vertex[y=-8.5,x=0,label=$y_2$]{7}
\Vertex[y=-6,x=7,label=$y_3$]{10}

  \Vertex[y=-12,x=0,label=$c$]{11}
  
  \Edge[Direct,label={1},bend=-5](1)(2)
\Edge[Direct,label={1}](1)(3)
\Edge[Direct,label={1},bend=-5](1)(5)
\Edge[Direct,label={1}](1)(6)
\Edge[Direct,label={1},bend=10](1)(8)
\Edge[Direct,label={1},bend=10](1)(9)

 \Edge[Direct,label={2},bend=-25](1)(2)
 \Edge[Direct,label={2},bend=20](1)(3)
 \Edge[Direct,label={2},bend=10](1)(5)
  \Edge[Direct,label={2},bend=20](1)(6)
   \Edge[Direct,label={2},bend=-5](1)(8)
    \Edge[Direct,label={2},bend=30](1)(9)

  \Edge[Direct,label={2}](2)(4)
  \Edge[Direct,label={2}](3)(4)
   \Edge[Direct,label={2}](5)(7)
  \Edge[Direct,label={2}](6)(7)
   \Edge[Direct,label={2}](8)(10)
  \Edge[Direct,label={2}](9)(10)

 \Edge[Direct,label={2},bend=-50](2)(11)
  \Edge[Direct,label={2},bend=30](6)(11)
   \Edge[Direct,label={2},bend=50](9)(11)
\end{tikzpicture}
\caption{All temporal arcs have elapsed time~1.}\label{fig:reductionLD1}
\end{subfigure}
    \caption{Example of the construction in the proof of Theorem \ref{thm:LD_min_hard}. Clause $c$ is equal to $(x_1\vee \neg x_2\vee\neg x_3)$. The value on top of each arc is the starting time.}
    \label{fig:LD_dsubgraph}
\end{figure}

\begin{theorem}\label{thm:MT_min_hard}
Let $\mc G=(V,A,\tau)$ be a temporal graph, $r\in V$ and $k\in \mathbb N$. Then deciding whether $\mc G$ has an \MT-\TSS with root $r$ and with at most $k$ temporal arcs is \emph{\NP}-complete even if $\tau=2$ and $el(a)= 0$ for every $a\in A$, or if $\tau=4$ and $el(a)= 1$ for every $a\in A$.
\end{theorem}

\begin{proof}
  Let $\mc G'=(V,A',\tau')$ be a temporal subgraph of $\mc G$. Computing $\MT\sub{\mc G}(r,v)$ and $\MT\sub{{\mc G'}}(r,v)$ for every vertex $v$ can be done in polynomial time (Table \ref{table:shortest_path}), as well as checking if $|A'|\leq k$, so the problem is in \NP.

\noindent To prove hardness we make a reduction from 3-SAT. 
Consider a formula $\phi$ in CNF form on variables $X = \{x_1,\ldots,x_l\}$ and on clauses $C = \{c_1,\ldots,c_m\}$. We construct a temporal graph $\mc G = (V,A,2)$ in the following way. Let $V =\bigcup_{i\in [l]} \{x^p_i, x^n_i, z^p_i, z^n_i,y_i\} \cup C\cup \{r\}$. Notice that $|V|=5l+m+1$.
For each $\alpha\in \{p,n\}$ and $i\in [l]$, add to $A$ the temporal arcs $(r,x^{\alpha}_i,2,2)$, $(r,z^{\alpha}_i,1,1)$, $(z^{\alpha}_i,x^{\alpha}_i,1,1)$ and $(x^{\alpha}_i,y_i,1,1)$. 
Then, for each clause $c_j$ and each variable $x_i$ appearing in $c_j$, add the temporal arc $(x^p_i,c_j,1,1)$ if $x_i$ appears in $c_j$ positively, while add the temporal arc $(x^n_i,c_j,1,1)$ if $x_i$ appears in $c_j$ negatively (see Figure \ref{fig:reductionMT0} as example). Observe that $\MT\sub{ \mc G}(r,x^{\alpha}_i)=1=\MT\sub{ \mc G}(r,z^{\alpha}_i)$ for all $\alpha\in \{p,n\}$, $i\in [l]$, and $\MT\sub{ \mc G}(r,y_i)=3=\MT\sub{ \mc G}(r,c_j)$ for all $i\in [l]$ and $j\in [m]$.
We now prove that $\phi$ is satisfiable if and only if there exists an $\MT$-\TSS of $\mc G$ with at most $6l+m $ temporal arcs.
Suppose first that $\phi$ has a satisfying assignment; we show how to construct an $\MT$-\TSS  $\mc G'=(V,A',2)$ with root $r$ of $\mc G$ with $|A'|\leq 6l+m $.
For each variable $x_i$, if $x_i$ is true then add to $A'$ the temporal arcs 
$$(r,x^p_i,2,2),\ (r,z^p_i,1,1),\ (z^p_i,x^p_i,1,1),\ (x^p_i,y_i,1,1),\ (r,x^n_i,2,2),\ (r,z^n_i,1,1),$$ while if $x_i$ is false then add to $A'$ the temporal arcs $$(r,x^n_i,2,2),\ (r,z^n_i,1,1),\ (z^n_i,x^n_i,1,1),\ (x^n_i,y_i,1,1),\ (r,x^p_i,2,2),\  (r,z^p_i,1,1).$$ 
Now consider a clause $c_j$ and choose one of the variables that validates $c_j$, say $x_{i_j}$. Add to $A'$ the unique temporal arc with head $c_j$ and tail $x^{\alpha_{i_j}}_{i_j}$, $\alpha_{i_j}\in \{p,n \}$. It holds that $|A'|= 6l+m $.
Observe that for all $v\in V$ there exists a temporal $(r,v)$-path in $\mc G'$ realizing $\MT\sub{\mc G}(r,v)$. 

\noindent Suppose now that ${\mc G'} = (V,A',2)$ is an $\MT$-\TSS of $\mc G$ with $|A'|\leq 6l+m $. Since $\mc G'$ is spanning, each vertex different from the root must have at least one incoming arc, so $|A'|\geq 5l+m=|V|-1 $. In particular, for all $i\in [l]$, $\alpha\in \{p,n \}$, the arcs $(r,x^\alpha_{i},2,2) $ and $(r,z^\alpha_{i},1,1) $ must belong to $ A'$ as they are the only ones realizing the distance for such vertices. Then notice that the only way a vertex $y_i$ can be temporally reachable from $r$ is that a least one of the arcs $(z^p_{i},x^p_{i},1,1)$ and $(z^n_{i},x^n_{i},1,1)$ belong to $A'$.
Since $|A'|\leq 6l+m $, it must hold that for every $i\in [l]$, exactly one of the arcs between $(z^p_{i},x^p_{i},1,1)$ and $(z^n_{i},x^n_{i},1,1)$ belongs to $A'$.
If the former case occurs, then set $x_i$ to true, while if the latter case occurs, then set $x_i$ to false. We now argue that this must be a satisfying assignment to $\phi$. For this, consider a clause $c_j$. Each clause $c_j$ must be reached exactly by one arc in $\mc G'$, as otherwise there would not be enough arcs to connect every vertex; let $(x^{\alpha_{i_j}}_{i_j},c_j,1,1)$ be this arc. This also implies that $(z^{\alpha_{i_j}}_{i_j},x^{\alpha_{i_j}}_{i_j},1,1)\in A'$, as otherwise $c_j$ would not be temporally reachable from $r$. 
If $\alpha_{i_j}=p$, then $x_{i_j} $ is set to true and,
by construction, $x_{i_j}$ appears positively in $c_j$; so $x_{i_j}$ satisfies $c_j$.
If $\alpha_{i_j}=n$, then $x_{i_j} $ is set to false and,
by construction, $x_{i_j}$ appears negatively in $c_j$; so $x_{i_j}$ satisfies $c_j$.

In the case where $el(a)=1$ for every arc $a$, the reduction is similar to the previous one. Specifically, for each $i\in [l]$, $\alpha\in \{p,n \}$, the temporal graph is made of the arcs $(r,x^{\alpha}_i,3,4)$, $(r,z^{\alpha}_i,1,2)$, $(z^{\alpha}_i,x^{\alpha}_i,2,3)$ and $(x^{\alpha}_i,y_i,3,4)$. Then for each clause $c_j$, if $x_i$ appears positively in $c_j$ we add the temporal arc $(x^p_i,c_j,3,4)$, while if $x_i$ appears negatively in $c_j$ we add the temporal arc $(x^n_i,c_j,3,4)$. See Figure \ref{fig:reductionMT1} as example. Analogous arguments to the previous ones apply.
\end{proof}
\begin{figure}
    \centering
 \begin{subfigure}[t]{0.45\textwidth}
   \begin{tikzpicture}[scale=0.28]
\SetVertexStyle[FillColor=white]
\SetEdgeStyle[Color=black,LineWidth=0.5pt]
  \Vertex[y=3,x=0.5,label=r,color=white!70!black]{1}
  

 \Vertex[y=-4,x=-8,label=$x_1^p$]{2}
  \Vertex[y=-4,x=-4.5,label=$x_1^n$]{3}
\Vertex[y=-5,x=-1,label=$x_2^p$]{5}
  \Vertex[y=-5,x=2,label=$x_2^n$]{6}
  \Vertex[y=-4,x=6,label=$x_3^p$]{8}
  \Vertex[y=-4,x=9,label=$x_3^n$]{9}

  \Vertex[y=-8,x=-6,label=$y_1$]{4}
\Vertex[y=-8.5,x=0.5,label=$y_2$]{7}
\Vertex[y=-8,x=6,label=$y_3$]{10}

  \Vertex[y=-12,x=0,label=$c$]{11}

\Vertex[y=0,x=-8,label=$z_1^p$]{12}
  \Vertex[y=0.5,x=-4.5,label=$z_1^n$]{13}
\Vertex[y=-1,x=-1,label=$z_2^p$]{14}
  \Vertex[y=-1,x=2,label=$z_2^n$]{15}
  \Vertex[y=0,x=6,label=$z_3^p$]{16}
  \Vertex[y=0,x=9,label=$z_3^n$]{17}

  \Edge[Direct,label={1},bend=-25](1)(12)
\Edge[Direct,label={1}](1)(13)
\Edge[Direct,label={1}](1)(14)
\Edge[Direct,label={1},bend=5](1)(15)
\Edge[Direct,label={1},bend=10](1)(16)
\Edge[Direct,label={1},bend=30](1)(17)

 \Edge[Direct,label={1},bend=0](12)(2)
\Edge[Direct,label={1},bend=0](13)(3)
\Edge[Direct,label={1},bend=0](14)(5)
\Edge[Direct,label={1}](15)(6)
\Edge[Direct,label={1},bend=0](16)(8)
\Edge[Direct,label={1},bend=0](17)(9)

 \Edge[Direct,label={2},bend=-35,distance=0.7](1)(2)
 \Edge[Direct,label={2},bend=-10](1)(3)
 \Edge[Direct,label={2},bend=15,distance=0.7](1)(5)
  \Edge[Direct,label={2},bend=45,distance=0.7](1)(6)
   \Edge[Direct,label={2},bend=20](1)(8)
    \Edge[Direct,label={2},bend=35,distance=0.8](1)(9)

  \Edge[Direct,label={1}](2)(4)
  \Edge[Direct,label={1}](3)(4)
   \Edge[Direct,label={1}](5)(7)
  \Edge[Direct,label={1}](6)(7)
   \Edge[Direct,label={1}](8)(10)
  \Edge[Direct,label={1}](9)(10)

 \Edge[Direct,label={1},bend=-50](2)(11)
  \Edge[Direct,label={1},bend=30](6)(11)
   \Edge[Direct,label={1},bend=50](9)(11)
 
\end{tikzpicture}
\caption{All temporal arcs have elapsed time~0.}\label{fig:reductionMT0}
\end{subfigure}
$ \quad$
 \begin{subfigure}[t]{0.45\textwidth}
   \begin{tikzpicture}[scale=0.28]
\SetVertexStyle[FillColor=white]
\SetEdgeStyle[Color=black,LineWidth=0.5pt]
  \Vertex[y=3,x=0.5,label=r,color=white!70!black]{1}
  

 \Vertex[y=-4,x=-8,label=$x_1^p$]{2}
  \Vertex[y=-4,x=-4.5,label=$x_1^n$]{3}
\Vertex[y=-5,x=-1,label=$x_2^p$]{5}
  \Vertex[y=-5,x=2,label=$x_2^n$]{6}
  \Vertex[y=-4,x=6,label=$x_3^p$]{8}
  \Vertex[y=-4,x=9,label=$x_3^n$]{9}

  \Vertex[y=-8,x=-6,label=$y_1$]{4}
\Vertex[y=-8.5,x=0.5,label=$y_2$]{7}
\Vertex[y=-8,x=6,label=$y_3$]{10}

  \Vertex[y=-12,x=0,label=$c$]{11}

\Vertex[y=0,x=-8,label=$z_1^p$]{12}
  \Vertex[y=0.5,x=-4.5,label=$z_1^n$]{13}
\Vertex[y=-1,x=-1,label=$z_2^p$]{14}
  \Vertex[y=-1,x=2,label=$z_2^n$]{15}
  \Vertex[y=0,x=6,label=$z_3^p$]{16}
  \Vertex[y=0,x=9,label=$z_3^n$]{17}

  \Edge[Direct,label={1},bend=-25](1)(12)
\Edge[Direct,label={1}](1)(13)
\Edge[Direct,label={1}](1)(14)
\Edge[Direct,label={1},bend=5](1)(15)
\Edge[Direct,label={1},bend=10](1)(16)
\Edge[Direct,label={1},bend=30](1)(17)

 \Edge[Direct,label={2},bend=0](12)(2)
\Edge[Direct,label={2},bend=0](13)(3)
\Edge[Direct,label={2},bend=0](14)(5)
\Edge[Direct,label={2}](15)(6)
\Edge[Direct,label={2},bend=0](16)(8)
\Edge[Direct,label={2},bend=0](17)(9)

 \Edge[Direct,label={3},bend=-35,distance=0.7](1)(2)
 \Edge[Direct,label={3},bend=-10](1)(3)
 \Edge[Direct,label={3},bend=15,distance=0.7](1)(5)
  \Edge[Direct,label={3},bend=45,distance=0.7](1)(6)
   \Edge[Direct,label={3},bend=20](1)(8)
    \Edge[Direct,label={3},bend=35,distance=0.8](1)(9)

  \Edge[Direct,label={3}](2)(4)
  \Edge[Direct,label={3}](3)(4)
   \Edge[Direct,label={3}](5)(7)
  \Edge[Direct,label={3}](6)(7)
   \Edge[Direct,label={3}](8)(10)
  \Edge[Direct,label={3}](9)(10)

 \Edge[Direct,label={3},bend=-50](2)(11)
  \Edge[Direct,label={3},bend=30](6)(11)
   \Edge[Direct,label={3},bend=50](9)(11)
 
\end{tikzpicture}
\caption{All temporal arcs have elapsed time~1.}\label{fig:reductionMT1}
\end{subfigure}
    \caption{Example of the construction in the proof of Theorem \ref{thm:MT_min_hard}. Clause $c$ is equal to $(x_1\vee \neg x_2\vee\neg x_3)$. The value on top of each arc is the starting time.}
    \label{fig:MT_dsubgraph}
\end{figure}
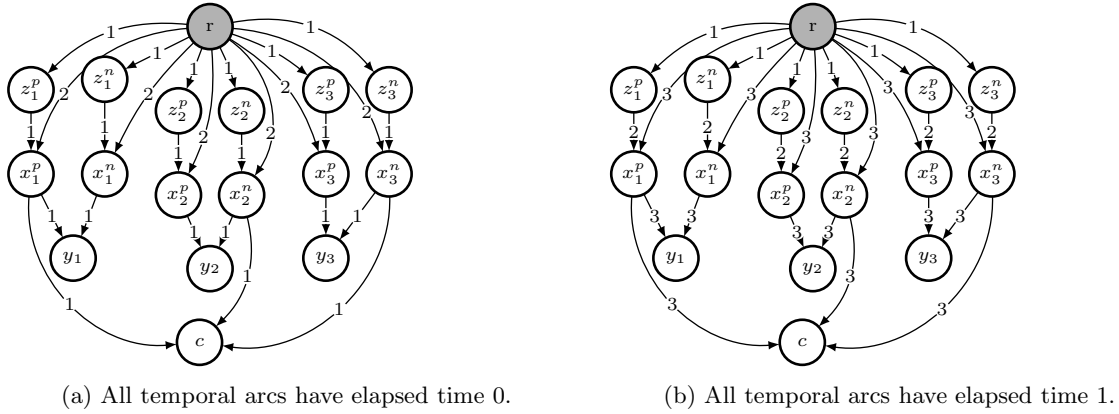
\noindent The gaps left by the above theorem are when $\tau=1$ or when $\tau\in\{2,3\}$ and all arcs have elapsed time at least~1. In the first case, the temporal graph reduces to a static graph, so the problem is solvable in polynomial time by Dijkstra's algorithm. When $\tau=2$ and all arcs have elapsed time at least~1, the minimum \MT-\TSS rooted in $r$ contains exactly all the arcs of type $(r,v,1,2)$, for every $v\in V$, if they exist; otherwise there is no \MT-\TSS of $\mc G$. When $\tau=3$ and all arcs have elapsed time at least~1, if $v$ is temporally reachable from $r$, then either $\MT(r,v)=1$ or $\MT(r,v)=2$. Notice that every vertex such that $\MT(r,v)=2$ has a \MT-prefix-optimal path from the root; hence every vertex temporally reachable from the root has a \MT-prefix-optimal path from $r$. Then, by Theorem \ref{thm:algMTST}, a maximum \MT-\TOB of $\mc G$ is computable in polynomial time. If it is also spanning, then it is in particular a minimum \MT-\TSS; if it is not spanning, then there is no \MT-\TSS of $\mc G$.

Theorems \ref{thm:TSS_hard}, \ref{thm:LD_min_hard}, and \ref{thm:MT_min_hard}, together with Theorem \ref{thm:FT_tob_hard} prove that finding a minimum $\textsc{d}$-\TSS of a temporal graph is an \NP-hard problem for each $\textsc{d}\neq$\EA.

\begin{remark}\label{rem:TISS}
Similarly to what done for \TOBs, we could define a $\textsc{d}$-Temporal In-Spanning Subgraph ($\textsc{d}$-\TISS) with root $r$  of a temporal graph $\mc G=(V,A,\tau)$ as the subgraph where, for all $ v\in V$, there exists a temporal $(v,r)$-walk realizing $ \textsc{d}\sub {\mc G}(r,v)$. If in addition it has the least number possible of temporal arcs, then we call it a \emph{minimum} $\textsc{d}$-\TISS. We then may ask what is the complexity of finding a minimum $\textsc{d}$-\TISS of a temporal graph. 
Thanks to the transformation $\cir$ presented in Definition \ref{def:reverse} and by following the lines of the proof of Proposition \ref{prop:equiv_tobtib}, it is easy to show that, given a subgraph $\mc G'$ of $\mc G$:
\begin{itemize}
       \item $\mc G'$ is a minimum \EA-\TISS of ${\mc G}$ if and only if $\mc G'^{\cir}$ is a minimum \LD-\TSS of $\mathcal{G}^{\circlearrowleft}$;
    \item $\mc G'$ is a minimum \LD-\TISS of ${\mc G}$ if and only if $\mc G'^{\cir}$ is a minimum \EA-\TSS of $\mathcal{G}^{\circlearrowleft}$;
    \item For each $\textsc{d}\!\in\! \{\FT,\MT,\MW,\ST \} $, $\mc G'$ is a minimum $\textsc{d}$-\TISS of ${\mc G}$ if and only if $\mc G'^{\cir}$ is a minimum $\textsc{d}$-\TSS of $\mathcal{G}^{\circlearrowleft}$.
\end{itemize}
This, together with Theorems \ref{thm:FT_tob_hard}, \ref{thm:TSS_hard}, \ref{thm:LD_min_hard}, and \ref{thm:MT_min_hard}, implies that for all $\textsc{d}\neq $ \LD, finding a minimum $\textsc{d}$-\TISS of a temporal graph is an \NP-hard problem, while a minimum \LD-\TISS is computable in polynomial time (Table \ref{tab:results}). 
\end{remark}

\section{Final remarks and conclusions}
We have showed that for \textsc{d} $\in\{$\LD,\MT,\ST$\}$, a spanning \textsc{d}-\TOB does not always exist, but computing a \textsc{d}-\TOB that spans the maximum number of vertices can be done in polynomial-time. Moreover, the overall complexity of the algorithms mainly depends on the complexity of computing the single source distances from the root to all the other vertices. 
In contrast, when \textsc{d} $\in\{$\FT, \MW$\}$, finding a spanning \textsc{d}-\TOB becomes \NP-complete. We then introduced a \textsc{d}-\TSS of a temporal graph as being a temporal subgraph containing a walk that realizes $\textsc{d}(r,v)$ for each vertex $v$, and we investigated the complexity of deciding the existence of a \textsc{d}-\TSS with at most $k$ arcs, for a given $k$. We showed that, for any distance $\textsc{d}\neq \EA$, finding such subgraph is an \NP-complete problem. Same results can be proven for  \textsc{d}-\TISS.
We highlight that all the hardness results of this paper (Theorems \ref{thm:FT_tob_hard}, \ref{thm:TSS_hard}, \ref{thm:LD_min_hard} and \ref{thm:MT_min_hard}) can be modified to meet the condition that the temporal graph in the input must have as underlying structure a digraph\footnote{Sometimes these temporal graphs are called \emph{simple}.} (instead of a multi-digraph), i.e.\ for every pair of vertices $u$ and $v$ we require at most one temporal arc having tail $u$ and head $v$. These modifications are summarized in the following remark.

\begin{remark}\label{rem:digraph}
List of the modifications to be implemented in the reductions of the corresponding theorems to meet the condition that the temporal graph in the input must have as underlying structure a digraph.\\
\emph{Theorem \ref{thm:FT_tob_hard}:} In the case $el(a)=0$ for all $a\in A$, set $V = X\cup \{y^1_1,\dots ,y^1_n\}\cup \{y^2_1,\dots ,y^2_n\} \cup C\cup \{r\}$, and for each variable $x_i$ add the temporal arcs $(r,y^1_i,1,1)$, $(y^1_i,x_i,1,1)$, $(r,y^2_i,2,2)$ and $(y^2_i,x_i,2,2)$. The temporal arcs connecting the variables to the clauses remain the same. 
In the case $el(a)=1$ for all $a\in A$, consider the same vertex set $V$ and add the temporal arcs $(r,y^1_i,1,2)$, $(y^1_i,x_i,2,3)$, $(r,y^2_i,2,3)$ and $(y^2_i,x_i,3,4)$, while augmenting by $1$ the starting and arrival times of the temporal arcs connecting the variables to the clauses.\\ 
\emph{Theorem \ref{thm:MT_min_hard}:} no modifications needed, the temporal graph built in the reduction is already simple.\\
\emph{Theorem \ref{thm:TSS_hard}:} subdivide the arcs connecting the root to a vertex. Specifically, replace each arc $(r,x^{\alpha}_i,3,3)$ by the arcs $(r,y^{\alpha}_i,3,3)$ and $(y^{\alpha}_i,x^{\alpha}_i,3,3)$, while the rest remains the same.\\
\emph{Theorem \ref{thm:LD_min_hard}:} In the case $el(a)=0$ for all $a\in A$, replace each arc $(r,x^{\alpha}_i,2,2)$ by the arcs $(r,y^{\alpha}_i,2,2)$ and $(y^{\alpha}_i,x^{\alpha}_i,2,2)$. In the case $el(a)=1$ for all $a\in A$, replace each arc $(r,x^{\alpha}_i,2,3)$ by the arcs $(r,y^{\alpha}_i,2,3)$ and $(y^{\alpha}_i,x^{\alpha}_i,3,4)$ and each arc $(r,x^{\alpha}_i,1,2)$ by the arcs $(r,z^{\alpha}_i,1,2)$ and $(y^{\alpha}_i,z^{\alpha}_i,2,3)$. All the other arcs have both their starting and arrival times increased by one.
\end{remark}

\noindent Notice that in the cases where the elapsed time are all equal to $1$, the above transformations make the temporal graph have lifetime $\tau=5$. We leave open the question whether the same complexity holds for simple temporal graphs with $\tau=4$.

The hardness results presented in this paper are clearly an issue when it comes to implement solutions for real life networks. In this direction, it would be of interest to study whether some of these networks, like public transport networks, present particular structures in their topology in order to restrict the analysis to a specific class of temporal graphs, for which the results might become tractable.




\bibliographystyle{plain}
\bibliography{bibliography}

\end{document}